%% file: control.tex
\documentclass[a4paper]{article}

\usepackage[margin=1in]{geometry} 

\usepackage{amsmath}
\usepackage{amsthm}
\usepackage{amssymb}

\usepackage[utf8]{inputenc}
\usepackage{hyperref}

\usepackage{wrapfig}

\theoremstyle{plain}
\newtheorem{theorem}{Theorem}
\newtheorem{corollary}[theorem]{Corollary}
\newtheorem{lemma}[theorem]{Lemma}
\newtheorem{proposition}[theorem]{Proposition}
\theoremstyle{definition}
\newtheorem{definition}[theorem]{Definition}

\newtheorem{remark}[theorem]{Remark}
\newtheorem{notation}[theorem]{Notation}

\usepackage{amsmath}
\makeatletter
\newcommand{\allignLabel}[1]{\refstepcounter{equation}(\theequation)\def\tmplab{#1}\ltx@label\tmplab}
\makeatother
\usepackage{anyfontsize}
\usepackage{proof}
\usepackage{mathrsfs}
\usepackage{stmaryrd}
\usepackage{dsfont}
\usepackage{cleveref}
\usepackage[all]{hypcap}
\usepackage{mathtools}
\usepackage[shortlabels]{enumitem}
\usepackage{quiver}
\usepackage{authblk}

\usepackage{tikzit}
\input{figures/quantum-circuits.tikzstyles}
\input{figures/quantum-circuits.tikzdefs}
\newcommand{\tf}[1]{\scalebox{0.90}{\input{#1.tikz}}}

\input{commands}

\title{One rig to control them all}
\author[1]{Chris Heunen}
\author[2]{Robin Kaarsgaard}
\author[1]{Louis Lemonnier}
\affil[1]{University of Edinburgh}
\affil[2]{University of Southern Denmark}
\date{}

\begin{document}

\maketitle

\begin{abstract}
	Controlled commands---computations whose execution depends on a separate
	input---play a central role in reversible Boolean circuits and quantum
	circuits. However, existing formalisms typically treat control only
	implicitly, entangled with other aspects of computation. From a semantic
	perspective, control is most naturally expressed in semisimple rig
	categories, which---unlike standard circuit models such as
	props---support both parallel and conditional composition.

	We present a construction that freely adjoins an explicit syntactic
	notion of control to a circuit theory specified as a suitable prop,
	subject to eight universally quantified equations. Our main result is
	that these equations are sound and complete for the intended semantics of
	control: the resulting theory satisfies a universal property, identifying
	it exactly as the circuit subtheory of the free semisimple rig
	completion. The proof combines coherence for rig categories with a new
	method based on induction over Gray codes.

	We illustrate the usefulness of the framework by showing that it
	simplifies several existing sound and complete axiomatisations of quantum
	circuits, isolating a small and conceptually clean set of generators and
	equations. In addition, the same equations yield a sound and complete
	axiomatisation of the multiply controlled Toffoli gate set, that is
	universal for reversible Boolean circuits.
\end{abstract}

\section{Introduction}

Many computations contain \emph{controlled} commands, that is, commands that are executed depending on the value of some memory cell.
By \emph{control} we mean the aspects of a computation that govern these dependencies.\footnote{We do not mean `control theory' as used in e.g.\ engineering~\cite{erbele:control}.}
Typically, the controlled command acts on one part of memory, and the controlling memory cell resides in another part of memory.
To be more precise, for example, consider controlled negation in reversible Boolean circuits: the target bit is flipped depending on the value of the control bit.
The goal of this article is to identify, separate, and study in isolation, this notion of computational control.

Traditional control flow or data flow is often mixed up with other computational aspects of circuits.
For example, in reversible Boolean or quantum circuits, multi-controlled gates such as the Toffoli gate are integral to universality and not treated differently than other, uncontrolled, gates~\cite{toffoli}.
Yet separating out the controlled aspects of a computation as specified by a circuit has several benefits.
\begin{enumerate}[(i)]
 \item \label{challenge:understand}
 Multi-controlled gates are of foundational importance in many computational theories, including Boolean logic~\cite{vollmer:circuits}, reversible computation~\cite{toffoli,thomsenkaarsgaardsoeken:ricercar}, and quantum computation~\cite{shendeetal:synthesis}. Isolating their control logic can help to better \emph{understand} these theories.
 \item \label{challenge:optimise}
 In quantum hardware, (multi-)controlled gates are among the most costly ones to perform physically~\cite{barencoetal:gates,yuduanying:toffoli,chenetal:control}. Separating out the control aspects can help find better optimisation strategies~\cite{balaucaarusaoie:simulating}.
 In general, partitioning off control aspects can help to \emph{optimise} computations in a generic way that is independent of the `base circuit theory' and therefore more efficient to apply.
 \item \label{challenge:simplify}
 Several recent results about logical completeness for quantum computation rely on elaborate families of equations~\cite{bianselinger:cliffordt,bianselinger:unitaries,lirossselinger:orthogonal,clementetal:extensions,clementetal:minimal,clementetal:complete,clementetal:coherentcontrol,clementperdrix:pbs,fangetal:hadamardpi}. Cordoning off control aspects can \emph{simplify}, and thereby clarify the core status of some and make them more modular.
\end{enumerate}

\begin{figure}
	\fbox{\begin{minipage}{0.99\textwidth}\begin{align*}
		\qquad \qquad \tf{monoidal-comp-2}
		& \;\stackrel{\ref{cprop:composition}}{=}\;
		\tf{monoidal-comp}
		&
		\tf{monoidal-id-2}
		& \;\stackrel{\ref{cprop:identity}}{=}\;
		\tf{monoidal-id}
		\\[2ex]
		\tf{monoidal-control-2}
		& \;\stackrel{\ref{cprop:strength}}{=}\;
		\tf{monoidal-control}
		&
		\tf{colour-change-2}
		& \;\stackrel{\ref{cprop:colour}}{=}\;
		\tf{colour-change}
		\\[2ex]
		\tf{complementarity}
		& \;\stackrel{\ref{cprop:complementarity}}{=}\;
		\tf{complementarity-2}
		&
		\tf{commuta}
		& \;\stackrel{\ref{cprop:commutativity}}{=}\;
		\tf{commuta-2}
		\\[2ex]
		\tf{swap-rule}
		& \;\stackrel{\ref{cprop:swap}}{=}\;
		\tf{swap-rule-2}
		&
		\tf{swap-coherence-2}
		& \;\stackrel{\ref{cprop:swapcoh}}{=}\;
		\tf{swap-coherence}
	\end{align*}\end{minipage}}
	\caption{Control equations.}
	\label{fig:control-equations}
\end{figure}

This article addresses these three challenges by introducing a theory of control governed by a handful of equations (see \Cref{fig:control-equations}). We argue that these equations completely capture control as follows.
\begin{enumerate}[(i)]
	\item
		The equations have clear computational interpretations, and several have appeared in the literature before~\cite{sharmaachour:optimizing,thomsenkaarsgaardsoeken:ricercar}.
		Additionally, we show that the equations are canonical in a strong way, by relating them to the natural mathematical notion of a \emph{rig category}~\cite{laplaza1972coherence,bimonoidal-book}.
		Starting with an arbitrary `base circuit theory', we syntactically construct a new `controlled circuit theory'. We do this in the most general setting possible, using \emph{prop}s~\cite{maclane:props,ghicajung:digitalcircuits,bonchisobocinskizanasi:hopf}. The construction has a universal property: roughly, the controlled prop is the free rig category on the base prop.
		This is borne out in examples: starting with the circuit theory consisting of a single NOT gate, the controlled theory is universal for reversible Boolean circuits. Starting with an additional Hadamard gate, the controlled theory is universal for quantum circuits~\cite{shi:aharonov}. Starting with a single $\sqrt{X}$ gate in fact suffices for the latter~\cite{kaarsgaard:ctrlv}.

	\item
		The coherence theorems for rig categories~\cite{laplaza1972coherence,bimonoidal-book}, and the fact that the control theory (of \Cref{fig:control-equations}) consists of only eight unquantified equations, give rise to many optimisation strategies.
		For example, we will show that the control theory suffices to syntactically derive equations like the following Sleator-Weinfurter decomposition quickly, that were before only known to hold semantically via matrix calculations in special cases~\cite{sleatorweinfurter,carette2024easy}:
		\[
			\tf{sl-w}
			=
			\tf{sl-w-8}
		\]

	\item The equations simplify related work. The first works on complete
		equational theories for quantum circuits~\cite{clementetal:complete,
		clementetal:minimal, clementetal:extensions} contain a notion of
		\emph{structural} equations, without a mathematical account of this
		notion. The same holds for~\cite{bianselinger:cliffordt,
		bianselinger:unitaries, lirossselinger:orthogonal}. 

		The closest work to ours is a sister notion of `controlled
		prop'~\cite{delorme2025control}, which is regarded as a property of
		a category rather than a construction: this article starts from a prop
		$\cat P$, and (under some conditions) constructs a prop $\ccat P$ that
		fits exactly the equations of a rig category, which is the result of a
		``controlled prop  construction'', similar to how a category can be
		completed with colimits.  On the other hand, \cite{delorme2025control} characterises categories that already
		support a ``control'' endofunctor that fits some equations, which also
		qualify as ``controlled props'', similarly to how $\cat{Set}$ verifies the property
		of having colimits. 
		Their setup covers a wider range of base theories, as it needs not
		consider NOT gates and colour changes: in their sense, our
		construction outputs a twice `controlled prop', with a distinguished
		involution to change between the two controls. Our approach allows for
		a strong connection between control and the structure of rig
		categories. 
	
		Similarly, the literature contains a set of template-based
		rewrite rules for control~\cite{thomsenkaarsgaardsoeken:ricercar}, but
		with missing equations (\emph{e.g.} \ref{cprop:swap} and
		\ref{cprop:swapcoh}) and therefore no proof of completeness.

		Our work also structures and elucidates a line of research on
		quantum programming languages taking semantics in rig
		categories~\cite{carette2016computing, heunen2022information,
		carette2023quantum, carette2024easy, fangetal:hadamardpi}. Finally, our
		results obviate work on circumventing no-control theorems by
		restricting to specific gate sets~\cite{bisioetal:conditional}.
\end{enumerate}

These results also substantiate the claim in the title, that rig structure encapsulates controlled computation, and only controlled computation. Thus rig categories form the bare minimum model of computation: the ability to compose instructions sequentially (with $\circ$), to consider data in parallel (with $\otimes$), and to use one piece of data to condition computations on another (using $\oplus$). This may also explain the ubiquity of matrices.

We proceed as follows. First, we reviews background material on props, rig categories, permutations, and Gray codes (in \Cref{sec:background}).
Then we introduce the main syntactic construction (in \Cref{sec:control}), and show that it semantically corresponds to rig completion (in \Cref{sec:rig}). The latter proof uses in a pleasingly structural way the new technique of \emph{Gray induction}, that proves a property for bitstrings by induction not on the normal successor of natural numbers, but by its Gray code.
After that, we discuss several applications of the theory (in \Cref{sec:examples}): a singly-generated universal theory for controlled reversible Boolean circuits, a generalised Sleator-Weinfurter decomposition~\cite{sleatorweinfurter}, a relationship to modal quantum theory~\cite{schumacherwestmoreland:modal}, and a universal singly-generated theory for controlled quantum circuits~\cite{kaarsgaard:ctrlv}, and complete equational theories for quantum circuits over various gate sets.
We conclude by indicating possible future work (in \Cref{sec:conclusion}).

\section{Background}
\label{sec:background}

\subsection{Props, rig categories, and bipermutative categories}
We assume familiarity with basic category theory, including (strict)
symmetric monoidal categories, and recall here some basic facts about
props, rig categories, and bipermutative categories. For a more in
depth treatment on these topics, the interested reader is invited to
consult, e.g., \cite{bimonoidal-book,maclane:props}.

A \emph{prop} is a strict symmetric monoidal category whose objects
are natural numbers, and where the monoidal product is given by
addition. A \emph{prop morphism} is an identity-on-objects strict
symmetric monoidal functor between props. Props are naturally regarded
as categories of diagrams, and play a similar role in the presentation
of symmetric monoidal categories as free groups do in combinatorial
group theory (see, e.g., \cite{bonchisobocinskizanasi:hopf}).

Rig categories are categories that are symmetric monoidal in two
different ways, such that one distributes over the other, similar to
how multiplication distributes over addition in arithmetic. More
formally, a rig category $(\cat{C},\otimes,\oplus,O,I)$ has two
symmetric monoidal structures $(\cat{C},\otimes,I)$ (the
\emph{(tensor) product}, with associator $\alpha^\otimes$, unitors
$\lambda^\otimes$ and $\rho^\otimes$, and symmetry $\symtimes$)
and $(\cat{C},\oplus,O)$ (the \emph{(direct) sum}, with associator
$\alpha^\oplus$, unitors $\lambda^\oplus$ and $\rho^\oplus$, and
symmetry $\symplus$), and come with natural coherence
isomorphisms witnessing distributivity on the right $\delta_R : (A
\oplus B) \otimes C \to (A \otimes C) \oplus (B \otimes C)$ and
$\delta^0_R : O \otimes A \cong O$ and analogously on the left,
subject to a significant number of coherence conditions (see
\cite{bimonoidal-book}). When it is clear from the context, we will
elide marking coherence isomorphisms with their monoidal structure,
and simply write, e.g., $\alpha$ rather than $\alpha^\oplus$.

Bipermutative categories are the strictified version of rig categories: the two
symmetric monoidal structures are strict, and right distributivity becomes an
equality (such that $(A \oplus B) \otimes C = (A \otimes C) \oplus (B \otimes
C)$), as do both nullary distributors (such that $A \otimes O = O = O \otimes
A$), subject to three coherence conditions (see~\cite{may,bimonoidal-book}).
That this is the right notion of strict rig category is shown by the coherence
theorem that every rig category is rig equivalent to a bipermutative
one~\cite[Prop. 3.5]{may}. We say that a bipermutative category is
\emph{semisimple} when its object are natural numbers, and the two strict
symmetric monoidal structures are multiplication and addition of natural
numbers on objects. Semisimple bipermutative categories are naturally regarded
as props by forgetting the tensor product $(\otimes,I)$.

\begin{proposition}
	\label{prop:powers}
	Given a semisimple bipermutative category $\cat C$ and a natural number
	$k$, the full subcategory of $\cat C$ whose objects are powers of $k$ is a
	prop $(\cat C_k, \otimes, 1)$.
\end{proposition}

\begin{remark}
	By using the proposition above to work on bits by choosing powers of $2$,
	we have a natural involution with the symmetry $\symplus_{1,1}$. If all
	morphisms are endomorphisms, then we obtain a controllable prop $(\cat C_2,
	\otimes, 1, \symplus_{1,1})$ as defined below (see
	Definition~\ref{def:controllable}).
\end{remark}

We provide the most natural example of such a category in the next section.

\subsection{Permutations}
An important example of a semisimple bipermutative category is the
category of permutations $\cat{Perm}$: its objects are natural
numbers, and morphisms $n \to n$ are permutations $[n] \to [n]$. The
\emph{direct sum} $\oplus$, on objects, simply sums the numbers, and
on morphisms, concatenates them.
\[
	f \oplus g =
	\left\{
		\begin{array}{ccll}
			[m+n] & \to & [m+n] & \\
			i & \mapsto & f(i) & \text{if } i<m\text, \\
			i & \mapsto & g(i-m)+m & \text{otherwise.}
		\end{array}
	\right.
\]
The symmetry $\symplus$ in this category is defined as follows.
\[
	\symplus_{m,n} =
	\left\{
		\begin{array}{ccll}
			[m+n] & \to & [n+m] & \\
			i & \mapsto & n+i & \text{if } i<m\text, \\
			i & \mapsto & i-m & \text{otherwise.}
		\end{array}
	\right.
\]
The \emph{tensor product} $\otimes$ multiplies the two numbers on objects, and acts on morphisms as follows.
\[
	f \otimes g =
	\left\{
		\begin{array}{ccll}
			[mn] & \to & [mn] & \\
			an+b & \mapsto & f(a)n + g(b) &
		\end{array}
	\right.
\]
The symmetry $\symtimes$ in this category is defined as:
\[
	\symtimes_{m,n} =
	\left\{
		\begin{array}{ccll}
			[mn] & \to & [nm] & \\
			an+b & \mapsto & bm+a & \\
		\end{array}
	\right.
\]
In fact, $\cat{Perm}$ is both the free prop and the free bipermutative
category. This means that the tensor product and its symmetries can entirely be
described in terms of $\oplus$ and $\symplus$~\cite[Lemma
S3]{caretteetal:free}. In particular, for all $f \colon m \to m$ and $g \colon
n \to n$:
\begin{align}
	\label{eq:kron}
	f \otimes g &=
	\symtimes_{n, m} \circ (\underbrace{ f \oplus
	 f \oplus \cdots \oplus f}_{n\text{ times}}) \circ
	\symtimes_{m,n} \circ
	(\underbrace{ g \oplus g \oplus
	\cdots \oplus g}_{m\text{ times}}) \\[1.5ex]
	&=
	(\underbrace{ g \oplus g \oplus
	\cdots \oplus g}_{m\text{ times}})
	\circ
	\symtimes_{n, m} \circ (\underbrace{ f \oplus
	 f \oplus \cdots \oplus f}_{n\text{ times}}) \circ
	\symtimes_{m,n}
\end{align}
Note that permutations are generated by transpositions
$\tau_i = (i \;\; i+1)$, subject to equations $\tau_i^2 = \iid$ and $(\tau_i \tau_{i+1})^3 = \iid$ and $\tau_i \tau_j = \tau_j \tau_i$ when $i+1<j$.

\subsection{Gray codes}
\label{sec:gray}

In this paper, we work with a similar but different way of generating
permutations with transpositions on Gray codes. Gray codes are a way of
ordering bit strings such that the transition from any bit string to its
successor consists of a single bit flip--i.e., such that the Hamming
distance between adjacent bit strings is $1$. Gray codes play an important
role in the synthesis of reversible circuits (see, e.g.,
\cite{nielsenchuang}).

For example, a Gray code of bit strings of length $3$ is shown in
Figure~\ref{fig:gray}.
\begin{wrapfigure}{r}{0.22\textwidth}
	\fbox{\begin{minipage}{0.185\textwidth}
		\centering$\begin{array}{c|c}
			\text{Binary} & \text{Gray} \\\hline
			000 & 000 \\
			001 & 001 \\
			010 & 011 \\
			011 & 010 \\
			100 & 110 \\
			101 & 111 \\
			110 & 101 \\
			111 & 100
		\end{array}$\end{minipage}}
	\caption{A Gray code on bit strings of length $3$.}
	\label{fig:gray}
\end{wrapfigure}

More formally, we define a function $h_n \colon [2^n] \to
\{0,1\}^n$ as:
\[
h_0(0) = \varepsilon
\qquad
h_n(i) =
\left\{
\begin{array}{ll}
 0 \cdot h_{n-1}(i) & \text{if } i < 2^{n-1}\text, \\
 1 \cdot h_{n-1}(2^n-1-i) & \text{otherwise.}
\end{array}
\right.
\]
We then fix \emph{Gray transpositions} $\theta_{n,i}$ to swap
$h_n(i)$ with $h_n(i+1)$. The two latter words have common prefix $w_n(i)$
and common suffix $w'_n(i)$.
We also define permutations $r_n \colon [2^n] \to [2^n]$ that translate
between the lexicographic order and the Gray order:
\[
r_0(0) = \iid_1
\qquad
r_n(i) =
\left\{
\begin{array}{ll}
 r_{n-1}(i) & \text{if } i < 2^{n-1}\text, \\
 2^{n-1} + r_{n-1}(2^n-1-i) & \text{otherwise.}
\end{array}
\right.
\]
In other words, the transposition $\theta_{n,i}$ swaps $r_n(i)$ and $r_n(i+1)$.

\begin{proposition}
	\label{prop:gray-transposition}
	For all $n$ and $i$, we have:
	\[
		\theta_{n,i} =
		r_n^{-1} \circ (\iid_i \oplus \symplus_{1,1} \oplus \iid_{2^n - 2 - i}) \circ r_n \text.
	\]
\end{proposition}

Since the translation permutations are defined by induction on the number of
(qu)bits, we then obtain an inductive definition for the Gray transpositions.
\begin{equation}\label{eq:ind-theta}
	\theta_{n+1,i} =
	\left\{
		\begin{array}{ll}
			\theta_{n,i} \oplus \iid_{2^n} & \text{if } i < 2^n - 1\text, \\[1.5ex]
			\symtimes_{2^n,2} \circ (\iid_{2^n} \oplus \symplus_{1,1} \oplus
			\iid_{2^n-2}) \circ \symtimes_{2,2^n}
			& \text{if } i = 2^n - 1\text, \\[1.5ex]
			\iid_{2^n} \oplus \theta_{n,2^{n+1} - 1 - i} & \text{otherwise.}
		\end{array}
	\right.
\end{equation}

\section{Controlled props}
\label{sec:control}
In this section, we will show how to freely adjoin control to a prop,
such that arbitrarily controlled versions of all morphisms of the
original prop are morphisms in the controlled prop. For this
to make sense, the prop must satisfy some lightweight requirements,
namely those of a \emph{controllable prop}.

\begin{definition}
	\label{def:controllable}
	A \emph{controllable prop} $(\cat P, +, 0, x)$, or \emph{crop} for short, is
	a prop $(\cat P, +, 0)$ whose morphisms are endomorphisms and in which one
	morphism $x \colon 1 \to 1$ is a distinguished involution. We sometimes
	refer to this involution as the \emph{NOT gate}. A crop morphism is a prop
	morphism of the underlying props that preserves the chosen involution.
\end{definition}

We focus on endomorphisms, because those fit our intended interpretation, our intended models, and equations~\ref{cprop:complementarity} and~\ref{cprop:commutativity} only make sense for endomorphisms. In principle, however, one could generalise to arbitrary morphisms when dropping those two equations.

\begin{notation}
	We write $\sigma$ for the symmetry of the monoidal structure $+$ in a prop,
	not to be confused with the symmetries $\symplus$ and $\symtimes$ in a
	bipermutative category.
\end{notation}

One important example of controllable prop is the prop $\cat X$ generated
only by the NOT gate $x$.

We introduce a construction on controllable props, that consists in adding a
control operator. In order to have readable equations and better presentation,
we allow ourselves to write both a \emph{positive} control (on the left) and a
\emph{negative} control (on the right). In fact, only one of them is necessary,
since the other can be obtained by conjugating with the NOT morphism
(see~\ref{cprop:colour} in Definition~\ref{def:cprop}).
\[
	C^1(f) \colon \tf{def-control}
	\qquad
	C^0(f) \colon \tf{def-control-neg}
\]
Both these operators are defined as functors. This means that they preserve
both composition and identities.
\[
	\tf{functor-control-compo-2} = \tf{functor-control-compo}
	\qquad
	\tf{functor-control-identity} = \tf{functor-control-identity-2}
\]
We represent the chosen involution $x \colon 1 \to 1$ with the usual NOT gate:
\(
	\tf{not} \text.
\)

Because props are also equipped with a symmetry for the tensor, we can also
obtain controlled operation with a control wire at the bottom.
\[
	\tf{def-control-bottom} \defeq \tf{def-control-bottom-2}
\]
In the rest of the paper, we draw a single wire to picture objects $n$.

\begin{definition}[Controlled prop]
	\label{def:cprop}
	Given an arbitrary controllable prop $(\cat P, +, 0, x)$, we extend it to a
	new prop with endofunctors $C^0$ and $C^1$ such that if $f \colon n \to n$,
	then $C^b(f) \colon 1+n \to 1+n$, and that, for all $n$ and $f, f_1, f_2
	\colon n \to n$, we have equations:
	\begin{enumerate}[(\alph*)]
		\item\label{cprop:composition} composition:
			$C^1(g \circ f) = C^1(g) \circ C^1(f)$;
		\item\label{cprop:identity} identity:
			$C^1(\iid_n) = \iid_{n+1}$;
		\item\label{cprop:strength} strength:
			$C^1(f + \iid_m) = C^1(f) + \iid_m$;
		\item\label{cprop:colour} colour change:
			$(x + \iid_n) \circ C^0(f) \circ (x + \iid_n) = C^1(f)$;
		\item\label{cprop:complementarity} complementarity:
			$C^0(f) \circ C^1(f) = \iid_1 + f$;
		\item\label{cprop:commutativity} commutativity:
			$C^0(f_1) \circ C^1(f_2) = C^1(f_2) \circ C^0(f_1)$;
		\item\label{cprop:swap} ``swap'':
			$C^1(x) \circ \sigma_{1,1} \circ C^1(x) \circ \sigma_{1,1} \circ C^1(x)
			= \sigma_{1,1}$;
		\item\label{cprop:swapcoh} ``swap'' coherence:
			$(\sigma_{1,1} + \iid_n) \circ C^1(C^1(f)) = C^1(C^1(f))
			\circ (\sigma_{1,1} + \iid_n)$.
	\end{enumerate}
	Figure~\ref{fig:control-equations} gives a diagrammatic account of
	the above equations. Write $\ccat P$ for this new prop, that we call
	the \emph{controlled prop} of $\cat P$. Note that $(\ccat P, +, 0, x)$ is
	itself a controllable prop.
\end{definition}

\begin{remark}
	The strength and composition equations above imply that, for all $f
	\colon m \to m$ and $g \colon n \to n$, we have:
	\[
		\tf{monoidal-control-3}
		= \tf{monoidal-control-4}
		= \tf{monoidal-control-5}
	\]
\end{remark}

Throughout the paper, we make extensive use of multi-controlled gates, such as
\[
	\tf{multicontrolled} = \tf{multicontrolled-2}
	\colon
	(\iid_2 + \sigma_{1,1})
	\circ C^1(C^0(C^0(x)))
	\circ (\iid_2 + \sigma_{1,1})\text,
\]
that we denote by $C^{10}_0(x)$ for readability, and which are well-defined
thanks to Equation~\ref{cprop:swapcoh}. The use of the latter equation is
implicit when working with circuits in the rest of the paper, since its role
restricts to diagram coherence for multi-controlled gates. More generally,
given two Boolean words $w$ and $w'$, we write $C^w_{w'}(f)$ for the gate $f$
controlled on $\abs w$ top wires, positively or negatively according to $w$,
and controlled on $\abs{w'}$ bottom wires, according to $w'$.

There are two forms of multiple controlled morphisms. Already in
$\ccat{\cat{P}}$, one can control morphisms from $\cat{P}$ on multiple wires.
In $\ccat{\ccat{\cat{P}}}$, one can additionally formally control morphisms
from $\ccat{\cat{P}}$.
In fact, the assignment $\mathbf{c}-$ is a monad on the category of crops and
crop morphisms. The monad multiplication flattens the latter type of multiple
control, but because the two forms of multiple control are distinct, the monad
is not idempotent. The monad unit maps morphisms to their uncontrolled
counterpart.

\begin{remark}
	In the definition for controllable props (see
	Definition~\ref{def:controllable}), we do not require extra conditions on
	the distinguished involution $x$. Therefore, one could choose it to be the
	identity $\iid_1$ on any prop. This does not break the construction of the
	corresponding controlled prop (see Definition~\ref{def:cprop}), but the
	latter is very degenerate, where the swap of wires is the identity because
	of \ref{cprop:swap} and all operations under the control commute because of
	\ref{cprop:colour} and \ref{cprop:commutativity}.
\end{remark}

\subsection{Relation to permutations}
The simplest example of a controlled prop is already quite meaningful, as it captures
all classical computing. To justify this statement, we prove that the controlled
prop $\ccat X$ of the prop $\cat X$ generated by a single involution is
isomorphic to the category of permutations between finite sets of cardinality a power of $2$.

\begin{theorem}
	\label{th:cx-iso}
	The controlled prop $\ccat X$ is crop isomorphic to $\cat{Perm}_2$.
\end{theorem}

To prove this, we define functors in both directions and show that they are
inverse of each other. Let us first define a crop morphism $\alpha \colon \ccat
X \to \cat{Perm}_2$.
\[
	\begin{array}{lcll}
		\alpha \colon & \ccat X & \to & \cat{Perm}_2 \\[1ex]
		& n & \mapsto & 2^n \\[1ex]
		& x & \mapsto & \symplus_{1,1} \\[1ex]
		& \sigma_{m,n} & \mapsto & \symtimes_{2^m,2^n} \\[1ex]
		& C^1(f) & \mapsto & \iid_{2^n} \oplus \alpha(f) \\[1ex]
		& C^0(f) & \mapsto & \alpha(f) \oplus \iid_{2^n} \\[1ex]
		& \iid_m + f & \mapsto & \iid_{2^m} \otimes \alpha(f)
	\end{array}
\]

\begin{lemma}
	\label{lem:alpha-wd}
	The crop morphism $\alpha \colon \ccat X \to \cat{Perm}_2$ is well-defined.
\end{lemma}

The proof consists in showing that any equation that holds in $\ccat X$ also
holds in $\cat{Perm}_2$ after the application of $\alpha$. In summary,
the equations \ref{cprop:composition} and \ref{cprop:identity} hold in
$\cat{Perm}_2$ because $\oplus$ is a (bi)functor; \ref{cprop:strength} follows
by distributivity coherence; \ref{cprop:colour} by naturality of the symmetry
$\symplus$; \ref{cprop:complementarity} and
\ref{cprop:commutativity} by bifunctoriality of $\oplus$; \ref{cprop:swap} by
naturality of $\symplus$ and associative coherence.

We now define a functor $\beta \colon \cat{Perm}_2 \to \ccat X$ which maps a
Gray code transposition $\theta_{n,i}$ to its associated multi-controlled NOT:
\[
	\begin{array}{lcll}
		\beta \colon & \cat{Perm}_2 & \to & \ccat X \\
		& 2^n & \mapsto & n \\
		& \theta_{n,i} & \mapsto & C^{w_n(i)}_{w'_n(i)}(x)
	\end{array}
\]
where $w_n(i)$ and $w'_n(i)$ are respectively the common prefix and suffix of
the $i^{\text{th}}$ word in the Gray code $h_n(i)$. In a Gray code order, two
successive Boolean words have only one bit that differs, for example $1011$ and
$1001$ when $n=4$ and $i=13$. The operation that permutes these two words is a
NOT gate on the third bit, triggered when the first bit is $1$, the second $0$
and the fourth $1$. This particular example corresponds to the multi-controlled
gate $C^{10}_1(x)$ in $\ccat X$:
\[
	\tf{multicontrolled-3}
\]

We use the induction formula for Gray code transpositions (\ref{eq:ind-theta})
and the fact that all permutations are generated by (Gray code) transpositions
to prove that $\beta \colon \cat{Perm}_2 \to \ccat X$ is a well-defined
functor, and is in fact a crop morphism.

\begin{lemma}
	\label{lem:beta-wd}
	The functor $\beta \colon \cat{Perm}_2 \to \ccat X$ is a well-defined crop
	morphism, and
	\[
		\beta(\iid_{2^n} \oplus f) = C^1(\beta(f)) \text,
		\qquad
		\beta(f \oplus \iid_{2^n}) = C^0(\beta(f)) \text,
	\]
	for all morphisms $f \colon 2^n \to 2^n$.
\end{lemma}
\begin{proof}
	We recall that any permutation $2^n \to 2^n$ is generated by Gray code
	transpositions $\theta_{n,i}$, with the equations below, respectively
	referred to as involution, Yang-Baxter~\cite{perketauyang:yangbaxter}, and
	commutativity.
	\[
		\begin{array}{c}
			\forall i < 2^n -1, \theta_{n,i}^2 = \iid_{2^n} \text,
			\qquad
			\forall i < 2^n -2, (\theta_{n,i}\theta_{n,i+1})^3 = \iid_{2^n} \text,
			\\[0.8ex]
			\forall i+1 < j, \theta_{n,i}\theta_{n,j} = \theta_{n,j}\theta_{n,i} \text.
		\end{array}
	\]

	The functor $\beta \colon \cat{Perm}_2 \to \ccat X$ is well-defined: the
	involution equation follows from \ref{cprop:composition} and
	\ref{cprop:identity}; Yang-Baxter is proven by induction, the base case on
	two (qu)bits is a direct consequence of \ref{cprop:swap}, alongside
	\ref{cprop:colour} and \ref{cprop:complementarity}, and the rest follows
	thanks to the induction formula for Gray codes and \ref{cprop:swapcoh}; the
	commutation is an easy consequence of \ref{cprop:commutativity} in
	Definition~\ref{def:cprop}.

	This functor is monoidal: first note that in $\cat{Perm}_2$, we have
	$\symtimes_{2,2} = \theta_{2,1} \circ \theta_{2,2} \circ \theta_{2,1}$ (and we have
	$h_2(1) = 01$, $h_2(2) = 11$ and $h_2(3) = 10$), whose image by $\beta$ is then:
	\begin{align*}
		&\ \beta(\theta_{2,1} \circ \theta_{2,2} \circ \theta_{2,1}) \\
		&= \beta(\theta_{2,1}) \circ \beta(\theta_{2,2}) \circ \beta(\theta_{2,1}) \\
		&= C_1(x) \circ C^1(x) \circ C_1(x) \\
		&\stackrel{\text{(def)}}{=} \sigma_{1,1} \circ C^1(x) \circ \sigma_{1,1} \circ
		C^1(x) \circ \sigma_{1,1} \circ C^1(x) \circ \sigma_{1,1}
		\\
		&\stackrel{\ref{cprop:swap}}{=} \sigma_{1,1} \circ
		\sigma_{1,1} \circ \sigma_{1,1} = \sigma_{1,1}
	\end{align*}
	thus $\beta$ preserves the symmetry of the tensor; additionally, we know that:
	\begin{align*}
		\iid_2 \otimes \theta_{n,i}
		&= \theta_{n,i} \oplus \theta_{n,i} \\
		&= (\theta_{n,i} \oplus \iid_{2^n}) \circ (\iid_{2^n} \oplus \theta_{n,i}) \\
		&\stackrel{(\ref{eq:ind-theta})}{=} \theta_{n+1,i} \circ \theta_{n+1,2^{n+1}-1-i}
	\end{align*}
	and therefore, since $i < 2^n - 2$, we have:
	\begin{align*}
		\beta( \iid_2 \otimes \theta_{n,i})
		&= \beta(\theta_{n+1,i}) \circ \beta(\theta_{n+1,2^{n+1}-1-i}) \\
		&= C^{w_{n+1}(i)}_{w'_{n+1}(i)}(x) \circ C^{w_{n+1}(2^{n+1}-1-i)}_{w'_{n+1}(2^{n+1}-1-i)}(x) \\
		&= C^0 \left(C^{w_n(i)}_{w'_n(i)}(x) \right) \circ C^1 \left(C^{w_n(i)}_{w'_n(i)}(x) \right) \\
		&\stackrel{\ref{cprop:complementarity}}{=} \iid_1 +
		C^{w_n(i)}_{w'_n(i)}(x) = \iid_1 + \beta(\theta_{n,i})
	\end{align*}
	which concludes that the functor $\beta$ is monoidal, since any
	permutations is generated by the $\theta_{n,i}$. Moreover, since we have
	$\theta_{1,0} = \symplus_{1,1}$, and since its image by $\beta$ is $x$, the
	functor $\beta$ preserves the NOT gate. Therefore, $\beta$ is a crop
	morphism.

	Additionally, we prove that for all $f$, we have $\beta(\iid_{2^n}
	\oplus f) = C^1(\beta(f))$. Once again, it suffices to prove it for the
	generators, therefore:
	\begin{align*}
		\beta(\iid \oplus \theta_{n,i})
		&\stackrel{(\ref{eq:ind-theta})}{=} \beta(\theta_{n+1,2^{n+1}-1-i}) \\
		&= C^{w_{n+1}(2^{n+1}-1-i)}_{w'_{n+1}(2^{n+1}-1-i)}(x) \\
		&= C^1 \left(C^{w_n(i)}_{w'_n(i)}(x) \right) \\
		&= C^1(\beta(f))
	\end{align*}
	and similarly, we have $\beta(f \oplus \iid_{2^n}) = C^0(\beta(f))$.
\end{proof}

This implies that $\beta$ is the inverse of $\alpha$, which
suffices to prove Theorem~\ref{th:cx-iso}.

\section{Rigged props thanks to Kronecker}
\label{sec:rig}

As shown in the previous section, in particular in the definition of the
functor $\alpha$, control is linked to the additive tensor $\oplus$, which we
loosely call \emph{direct sum}. This means that if we can extend a `base
circuit theory' with direct sums, we should get its controlled
circuit theory. In fact, the better way is to start from a prop generated by
the direct sum, and then recapture the `base circuit theory' on the tensor
product $\otimes$ from there.

In the context of matrices, the Kronecker product can be computed with direct
sums only. Given a matrix $M$, the matrix $I_n \otimes M$ is obtained as the
block-diagonal matrix
\[
	\left.
		\begin{bmatrix}
			M && O \\
			& \ddots & \\
			O && M
		\end{bmatrix}
	\right\}
	n\text{ times}\text,
\]
or more precisely, $M \oplus \dots \oplus M$. Additionally, the
symmetry for the tensor product $\otimes$ can be expressed in terms of composition
of symmetries for the direct sum. This shows that any theory around the direct
sum can simulate the one of a tensor product, given some coherence. One of
these is the bifunctorality of the tensor product, or in pictures, the
following equality:
\[
	\tf{bifunctor}
	=
	\tf{bifunctor-2}
\]
In a bipermutative category, with $f \colon k \to l$ and $g \colon m \to
n$, we have:
\[
	\begin{array}{ll}
		f \otimes g
		&= (f \otimes \iid_n) \circ (\iid_k \otimes g) \\[.8ex]
		&= \symtimes_{n,l} \circ (\iid_n \otimes f) \circ \symtimes_{k,n} \circ (\iid_k \otimes g) \\[.8ex]
		&= \symtimes_{n,l} \circ (\underbrace{f \oplus \dots \oplus f}_{n\text{
			times}}) \circ \symtimes_{k,n} \circ (\underbrace{g \oplus \dots \oplus
		g}_{k\text{ times}}) \\[2.5ex]
		&= (\iid_l \otimes g) \circ (f \otimes \iid_m) \\[.8ex]
		&= (\iid_l \otimes g) \circ \symtimes_{m,l} \circ (\iid_m \otimes f) \circ \symtimes_{k,m} \\[.8ex]
		&= (\underbrace{g \oplus \dots \oplus g}_{l\text{ times}}) \circ
		\symtimes_{m,l} \circ (\underbrace{f \oplus \dots \oplus f}_{m\text{ times}})
		\circ \symtimes_{k,m} \\[.8ex]
	\end{array}
\]
To reflect this in a tensor product $\otimes$ generated by a direct sum
$\oplus$, the last line above needs to be enforced, leading to our definition of
rigged prop below.

\subsection{Rigged props}

We are now able to define rigged props as described above: a new prop whose
monoidal structure is meant as a direct sum. A prop standardly admits a
presentation with a set of generators and a set of
equations~\cite{baezetal:props}, similarly to groups. The generators of the
prop we construct are the ones from the base prop, for which we adjust the
domains and codomains to a matrix point of view. The relations for this new
prop are the ones of the base prop, in addition to the
relations~\eqref{eq:mon-tensor} that ensure that the newly obtained tensor
product is bifunctorial.

\begin{definition}[Rigged prop]
	\label{def:kprop}
	Given a prop $(\cat P,+,0)$ with a set of generators $G$, and a set of
	relations $R$, its \emph{rigged prop} $(\kcat P,\oplus, 0)$ is the prop
	generated by
	\begin{equation}\label{eq:kgen}
		\widetilde G = \left\{ \widetilde g \colon 2^k \to 2^l \alt g \colon k \to l \in G \right\}
	\end{equation}
 with equations
	\begin{align}
		\label{eq:mon-tensor}
		\widetilde R =~&
		\left\{ \widetilde{f + \iid} \circ \widetilde{\iid + h} = \widetilde{\iid + h} \circ
		\widetilde{f + \iid} \alt f,h \in G \right\} \\
		\label{eq:base-rel}
		&\cup
		\left\{ \widetilde f = \widetilde h \alt (f = h) \in R \right\}
	\end{align}
	where $\widetilde \cdot$ is defined as follows for $f \colon k \to l$, $g \colon l \to m$,
	and $h \colon m \to n$.
	\begin{align}
		\widetilde{\iid_n} &= \iid_{2^n} \\
		\widetilde{g \circ f} &= \widetilde g \circ \widetilde f \\
		\label{eq:new-tensor}
		\widetilde f \otimes \widetilde h \defeq \widetilde{f + h} &= \symtimes_{2^n, 2^l} {\circ}
		(\underbrace{\widetilde f \oplus \widetilde f \oplus \cdots \oplus \widetilde f}_{2^n\text{
			times}}) {\circ} \symtimes_{2^k,2^n} {\circ} (\underbrace{\widetilde h \oplus \widetilde h
		\oplus \cdots \oplus \widetilde h}_{2^k\text{ times}})
	\end{align}
\end{definition}

\begin{theorem}
	\label{th:krig}
	Given a prop $\cat P$, its rigged prop $\kcat P$ is semisimple
	bipermutative.
\end{theorem}
 \begin{proof}
	There is a prop morphism $E \colon \cat{Perm} \to \kcat P$ since
	$\cat{Perm}$ is the free prop. The $\otimes$ monoidal structure on $\kcat
	P$ can then be defined the same way as in $\cat{Perm}$. The tensor
	$\otimes$ in $\kcat P$ is indeed a bifunctor, mainly thanks to
	(\ref{eq:mon-tensor}). Equation~\eqref{eq:new-tensor} ensures that the
	symmetry of the tensor $\oplus$ is natural. The coherence equations are
	inherited from $\cat{Perm}$ through $E$. Therefore, $\kcat P$ is
	semisimple bipermutative.
 \end{proof}

This construction $\kcat \cdot$ provides not merely a semisimple bipermutative
category, but the universal one.

\begin{theorem}
	\label{th:kuniv}
	Given a prop $\cat P$ and a strict monoidal functor to a semisimple
	bipermutative category $F \colon \cat P \to \cat Q$ such that $F(1) = 2$,
	there is a unique strict bipermutative functor $\widetilde F \colon \kcat P
	\to \cat Q$ such that the following diagram commutes:
	\[
		\begin{tikzcd}[column sep=large]
			\cat P & \kcat P \\
			& \cat Q
			\arrow[hook, from=1-1, to=1-2]
			\arrow["F"', from=1-1, to=2-2]
			\arrow[dashed, "\widetilde F", from=1-2, to=2-2]
		\end{tikzcd}
	\]
\end{theorem}
\begin{proof}
	Let us first define the resulting functor. Since $F(1)=2$, we have
	$F(n) = 2^n$. Thus we can define:
	\[
		\begin{array}{c}
			\widetilde F (n) = n
			\qquad
			\widetilde F(\widetilde g) = F(g)
		\end{array}
	\]
	and this suffices to define the functor $\widetilde F$. All the equations
	in $\widetilde P$ derive from $P$ or from the rig structure, and therefore
	$\widetilde F$ is well-defined because it is strict rig and because $F$ is
	well-defined.
	A strict rig functor between two semisimple bipermutative categories is
	necessarily the identity on objects. Since the functor is obtained by the image
	of $F$, it is necessarily unique.
\end{proof}

The domains and codomains of the generators for $\kcat P$, in \eqref{eq:kgen},
can in fact be the power of any arbitrary natural number. We choose $2$ to fit
the story of circuits on bits. Moreover, the definition of rigged prop does not
account for the NOT gate, which plays an central role in our theory of control.
This is the focus of the next section.

\subsection{Rigged crops}

We now present a definition of rigged props that takes into account the choice
of NOT gate. Whether it is classical reversible circuits or quantum circuits,
the NOT gate is represented as the following matrix:
\[
	\begin{bmatrix}
		0 & 1 \\
		1 & 0
	\end{bmatrix}
\]
which is also the symmetry $\symplus_{1,1}$ for the direct sum $\oplus$. Given a
crop with a choice of NOT gate, we further identify the chosen NOT gate with
the symmetry $\symplus_{1,1}$ in its constructed rigged prop, to obtain its
rigged crop.

\begin{definition}[Rigged crop]
	Given a controllable prop $(\cat P, +, 0, x)$ with a set of generators $G$,
	and a set of relations $R$, its \emph{rigged crop} $(\scat P, \oplus, 0)$
	is the prop generated by
	\[
		\ovl G = \widetilde G
	\]
	subject to equations
	\[
		\ovl R =
		\left\{ \widetilde x = \symplus_{1,1} \right\}
		\cup
		\widetilde R\text.
	\]
\end{definition}

The rigged crop is semisimple bipermutative in the same way as the
rigged prop.

\begin{theorem}
	The category $(\scat P, \otimes, 1, \oplus, 0)$ is semisimple
	bipermutative.
\end{theorem}

Furthermore, rigged crops enjoy a similar universality as rigged props.

\begin{theorem}
	\label{th:suniv}
	Given a controllable prop $(\cat P, +, 0, x)$, and a strict monoidal
	functor $F \colon \cat P \to \cat Q$ to a semisimple bipermutative category
	with $F(x) = \symplus_{1,1}$, there is a unique strict rig functor $\ovl F
	\colon \scat P \to \cat Q$ such that the following diagram commutes.
	\[
		\begin{tikzcd}[column sep=large]
			\cat P & \scat P \\
			& \cat Q
			\arrow[from=1-1, to=1-2]
			\arrow["F"', from=1-1, to=2-2]
			\arrow[dashed, "\ovl F", from=1-2, to=2-2]
		\end{tikzcd}
	\]
\end{theorem}

In the following, we are interested in objects in $\scat P$ that are powers of
$2$. 

\begin{corollary}
	\label{cor:suniv}
	Given a controllable prop $(\cat P, +, 0, x)$, and a strict monoidal
	functor $F \colon \cat P \to \cat Q$ to a semisimple bipermutative category
	with $F(x) = \symplus_{1,1}$, there is a unique prop morphism $\ovl F_2
	\colon \scat P_2 \to \cat Q_2$ such that $\ovl F_2$ preserves the NOT gate and
	the following diagram commutes.
	\[
		\begin{tikzcd}[column sep=large]
			\cat P & \scat P_2 \\
			& \cat Q_2
			\arrow[from=1-1, to=1-2]
			\arrow["F"', from=1-1, to=2-2]
			\arrow[dashed, "\ovl F_2", from=1-2, to=2-2]
		\end{tikzcd}
	\]
\end{corollary}

\subsection{Relation to permutations}
Observe that, similarly to the controlled prop on a single involution, the
rigged crop on a single involution is the category of permutations.

\begin{theorem}
	The category $\scat X$ is exactly the category of permutations.
\end{theorem}
\begin{proof}
	Permutations are generated by transpositions $\tau_i = (i \;\; i+1)$, which can
	be mapped in $\scat X$ to the involutions $\iid_i \oplus \symplus_{1,1} \oplus \iid_{n-3-i}$.
	The latter commute if they act on separate parts of the
	object. Moreover, the Yang-Baxter equation $(\tau_i\;\; \tau_{i+1})^3	= \iid$ follows from naturality of $\symplus$ and associative coherence.
	Thus the functor $\cat{Perm} \to \scat{X}$ is well-defined and the identity on objects.
 Because it
	maps coherence morphisms to coherence morphisms, it is in fact a strict
	bipermutative functor. As $\scat X$, just like $\cat{Perm}$, only
	contains coherence morphisms, the functor is an isomorphism.
\end{proof}

We now redefine the functors $\alpha \colon \ccat X \to \tcat X$ and $\beta
\colon \tcat X \to \ccat X$ to fit the topic at hand.

\begin{remark}
	In principle, the notion of control makes sense for any prop consisting
	only of endomorphisms. However, we focus here on groupoids, since this is
	where controlled gates originated. In the more general, irreversible
	setting, control certainly makes sense, though more space-efficient gadgets
	such as multiplexers are typically used instead. Note that none of the
	control equations imply reversibility.
\end{remark}

\subsection{Rig is control}
\label{sec:finale}

In this section, we fix a controllable prop $(\cat P, +, 0, x)$. It is now
a matter of proving that $\ccat P$ and $\scat P_2$ are isomorphic. First let
us define a functor $\ccat P$ to $\scat P_2$.
\[
	\begin{array}{lrcl}
		A \colon & \ccat P & \to & \scat P_2 \\
		& n & \mapsto & 2^n \\
		& f \in \ccat X & \mapsto & \alpha(f) \\
		& g \in G & \mapsto & \ovl g \\
		& C^1(f) & \mapsto & \iid_{2^n} \oplus A(f) \\
		& C^0(f) & \mapsto & A(f) \oplus \iid_{2^n} \\
		& \iid_m + f & \mapsto & \iid_{2^m} \otimes A(f)
	\end{array}
\]

\begin{lemma}\label{lem:a-def}
	The prop morphism $A \colon \ccat P \to \tcat P$ is well-defined.
\end{lemma}
\begin{proof}
	We need to prove that all equations that hold in $\ccat P$ also hold in
	$\tcat P$ after taking the image under $A$. The equations of monoidal
	coherence hold because $\tcat P$ is a prop. We further observe that if $f$
	does not contain any control, then $A(f) = \widetilde f$; therefore, all
 original equations from $\cat P$ hold by~\eqref{eq:base-rel}.
	Additionally, the equations \ref{cprop:composition} and
	\ref{cprop:identity} hold in $\tcat P$ because $\oplus$ is a
	(bi)functor; \ref{cprop:strength} follows by coherence of distributivity;
	\ref{cprop:colour} by naturality of the symmetry $\symplus$;
	\ref{cprop:complementarity} and \ref{cprop:commutativity} by
	bifunctoriality of $\oplus$; \ref{cprop:swap} by naturality of $\symplus$
	and associative coherence; and~\ref{cprop:swapcoh} follows from
	associativity of $\oplus$, since it is a bracketing matter for control.
\end{proof}

Similarly, we can define a functor $\scat P_2 \to \ccat P$.
\[
	\begin{array}{lrcl}
		B \colon & \scat P_2 & \to & \ccat P \\
		& 2^n & \mapsto & n \\
		& f \in \scat X_2 & \mapsto & \beta(f) \\
		& \ovl g \in \ovl G & \mapsto & g \\
		& \iid_{2^n} \oplus f & \mapsto & C^1(B(f)) \\
		& f \oplus \iid_{2^n} & \mapsto & C^0(B(f))
	\end{array}
\]

\begin{lemma}\label{lem:b-def}
	The functor $B \colon \tcat P \to \ccat P$ is a well-defined prop morphism.
\end{lemma}
\begin{proof}
	Note that $\scat P$ is generated by $\ovl G$ and permutations
	$\symplus$. Therefore, morphisms are of the form:
	\[
		\gamma_0 \circ (\iid \oplus \widetilde{g_1} \oplus \iid) \circ
		\gamma_1 \circ \dots \circ \gamma_{n-1} \circ (\iid \oplus
		\widetilde{g_n} \oplus \iid) \circ \gamma_n
	\]
	where $\gamma_i$ are short for compositions of symmetries for the direct
	sum $\oplus$. Since the generator terms are surrounded by permutations, we
	can write them as follows, without loss of generality:
	\[
		\gamma_0 \circ (\iid \oplus \widetilde{g_1}) \circ \gamma_1 \circ \dots
		\circ \gamma_{n-1} \circ (\iid \oplus \widetilde{g_n}) \circ \gamma_n
	\]
	If this is a morphism $2^m \to 2^m$ in $\tcat P$, then the terms containing
	generators $\widetilde g \colon 2^k \to 2^k$ are of the form:
	\[
		\iid_{2^m - 2^k} \oplus \widetilde g
		=
		\iid_{2^{m-1}} \oplus (\iid_{2^{m-2}} \oplus (\dots \oplus (\iid_{2^k} \oplus \widetilde g) \dots ))
	\]
	Therefore the definition of $B$ above defines an image of all morphism in
	$\tcat P$. We now prove that $B$ is well-defined. All coherence equations
	hold because $\beta$ is well-defined. The only equation specific to $\scat
	P_2$ is (\ref{eq:mon-tensor}), which is preserved since $\otimes$
	is a monoidal structure, provided that $B$ is indeed a prop morphism. We
	thus need to prove that $B(f \otimes g) =
	B(f) + B(g)$ for all $f$ and $g$. We know, through $\beta$, that $B(\symtimes_{2^m,2^n}) =
	\sigma_{m,n}$, so that it suffices to show that $B(\iid_2 \otimes
	f) = \iid_1 + B(f)$.
	\begin{align*}
		B(\iid_2 \otimes f)
		&= B(f \oplus f) \\
		&= B(f \oplus \iid_n) \circ B(\iid_n \oplus f) \\
		&= C^0(B(f)) \circ C^1(B(f)) \\
		&\stackrel{\ref{cprop:complementarity}}{=} \iid_1 + B(f)
	\end{align*}
	Thus $B$ is a prop morphism.
\end{proof}

Both functors are well-defined and each other's inverse, yielding an isomorphism of props.

\begin{theorem}
	\label{th:cp-iso}
	The props $\ccat P$ and $\tcat P$ are isomorphic.
\end{theorem}

Thus the control equations of Figure~\ref{fig:control-equations} are exactly
the ones that turn a crop into a rigged crop. In this sense, we conclude they
are structural: they govern exactly the rig structure.  A fortiori,
Corollary~\ref{cor:suniv} shows that the control equations are minimal: our
definition of controlled prop (see Definition~\ref{def:cprop}) is universal
among such structural categories.

This main theorem explains the ubiquity of rig structure in computing: it
captures exactly the nature of control flow. The control equations have clear
computational interpretations, and the theorem shows that they are canonical in
a strong way. Conversely, the theorem bestows the mathematically pleasing
structure of rig categories with computational aspects.

\section{Applications}
\label{sec:examples}

To showcase the usefulness of the control equations of Figure~\ref{fig:control-equations}, this section discusses five applications drawn from circuit theory, both Boolean and quantum, that can be handled easily using our control theory.

\subsection{Reversible Boolean circuits}

We have already seen that controlling the degenerate prop consisting of a single NOT gate generates the full prop of permutations between sets of size power of two, that is, of reversible Boolean circuits. To make this more concrete, let us derive the circuit identity
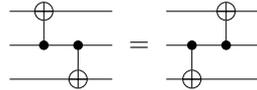
\[
 \tf{classical-example}
 =
 \tf{classical-example-4}
\]
expressing that two NOT gates that are controlled on the same wire but have different targets always commute. This may seem obvious, but usual proofs resort to truth tables. In contrast, the proof below is equational, and demonstrates how the control equations interact. First, observe the following.
\begin{equation}\label{eq:cl-example-lemma}
	\tf{classical-example-lemma}
	\stackrel{\text{(inv.)}}{=}
	\tf{classical-example-lemma-1}
	\stackrel{\text{\ref{cprop:complementarity}}}{=}
	\tf{classical-example-lemma-2}
	\stackrel{\text{\ref{cprop:composition}}}{=}
	\tf{classical-example-lemma-3}
\end{equation}
The desired equation now follows.
\begin{equation}\label{eq:cl-example}
	\tf{classical-example}
	\stackrel{\eqref{eq:cl-example-lemma}}{=}
	\tf{classical-example-2}
	\stackrel{\text{\ref{cprop:commutativity}}}{=}
	\tf{classical-example-3}
	\stackrel{\eqref{eq:cl-example-lemma}}{=}
	\tf{classical-example-4}
\end{equation}

\subsection{Sleator-Weinfurter}

The Sleator-Weinfurter construction is a general construction for
forming a doubly controlled $f^2$-gate using nothing but controlled
$f$-gates and controlled NOT gates (see Figure~\ref{fig:sl-w}). This
property has led to the quantum synthesis of reversible circuits being
studied by means of the NCV gate set $\{NOT,CV,CNOT\}$ (see, e.g.,
\cite{abdessaied2016complexity}). Originally shown through linear
algebra~\cite{barencoetal:gates}, the construction has a simple proof
in terms of the control equations.

\begin{figure*}
	\[
		\begin{array}{ll}
			\tf{sl-w} &
			\stackrel{\text{\ref{cprop:complementarity}}}{=} \tf{sl-w-2}
			\stackrel{\text{\ref{cprop:colour}}}{=} \tf{sl-w-3}
			\\[6ex]
			&
			\stackrel{\text{\ref{cprop:commutativity}
			\& \ref{cprop:composition}}}{=} \tf{sl-w-4} 
			\stackrel{\text{\ref{cprop:commutativity}}}{=} \tf{sl-w-5}
			\\[6ex]
			&
			\stackrel{\text{\ref{cprop:composition}}}{=} \tf{sl-w-6}
			\stackrel{\text{\ref{cprop:colour}}}{=} \tf{sl-w-7}
			\stackrel{\text{\ref{cprop:composition}}}{=} \tf{sl-w-8}
		\end{array}
	\]
	\caption{Sleator-Weinfurter identity through control equations.}
	\label{fig:sl-w}
\end{figure*}

\subsection{Schumacher-Westmoreland: modal quantum theory}

Starting from a theory of matrices, it is now easy to generate a circuit
equational theory with our control completion of circuit props. We give an
example of this on modal quantum theory~\cite{schumacherwestmoreland:modal}.
The latter is a toy theory for quantum mechanics where scalars are not complex
numbers but drawn from a finite field, for example the two-element field
$\Z_2$. Of course this cannot express all of quantum theory, but it
still retains key notions such as entanglement, mixed states, superdense coding
and teleportation~\cite{schumacherwestmoreland:modal}. This makes modal quantum
theory quite a rich toy model, which makes modal quantum circuits, acting on \emph{mobits} (modal bits), an interesting case study.

As pointed out by Lafont~\cite{lafont2003}, reversible matrices on $\Z_2$ have
two generators
\[
	X = \begin{bmatrix} 0 & 1 \\ 1 & 0 \end{bmatrix}
	\qquad
	J = \begin{bmatrix} 1 & 1 \\ 1 & 0 \end{bmatrix}
\]
governed by the following complete~\cite[Theorem 6]{lafont2003} set of
matrix equations.
\begin{align}
	\label{laf:xinv} X^2 &= \iid_2 \\
	\label{laf:qinv} J^3 &= \iid_2 \\
	\label{laf:bit} JXJ &= X \\
	\label{laf:yb} (\iid_1 {\oplus} X)(X {\oplus} \iid_1)(\iid_1 {\oplus} X)
	&= (X {\oplus} \iid_1)(\iid_1 {\oplus} X)(X {\oplus} \iid_1) \\
	\label{laf:ybq} (\iid_1 {\oplus} X)(J {\oplus} \iid_1)(\iid_1 {\oplus} X)
	&= (X {\oplus} \iid_1)(\iid_1 {\oplus} J)(X {\oplus} \iid_1) \\
	\label{laf:ybqqq} (\iid_1 {\oplus} J)(X {\oplus} \iid_1)(\iid_1 {\oplus} J)
	&= (J {\oplus} \iid_1)(\iid_1 {\oplus} J)(J {\oplus} \iid_1)
\end{align}
The first three equations are simple mobit equations, \eqref{laf:yb} is
Yang-Baxter, and \eqref{laf:ybq} is the naturality of the symmetry. The last one is
not structural, and translates to the following circuit equation.
\begin{equation}
	\label{eq:mobit-missing}
	\tf{mobit}
	=
	\tf{mobit-2}
\end{equation}

\begin{proposition}
	Equation \eqref{eq:mobit-missing} is independent of the control equations
	of Figure~\ref{fig:control-equations} and equations \eqref{laf:xinv},
	\eqref{laf:qinv}, and \eqref{laf:bit}.
\end{proposition}
\begin{proof}
	Interpret diagrams in $X$ and $J$ in permutation matrices, by mapping
	$X$ to the usual NOT gate and $J$ to the identity. All control equations
	and \eqref{laf:xinv}, \eqref{laf:qinv}, and \eqref{laf:bit} hold directly.
	However, equation~\eqref{eq:mobit-missing} does not hold in this model: the
	left-hand side is a swap, whereas the right-hand side is
	the identity.
\end{proof}

We show that this independent equation~\eqref{eq:mobit-missing} is the unique
one necessary for completeness in addition to the mobit equations and control
equations.

\begin{theorem}
	The control equations of Figure~\ref{fig:control-equations} and equations
	\eqref{laf:xinv}, \eqref{laf:qinv}, \eqref{laf:bit} and
	\eqref{eq:mobit-missing} are complete for mobit circuits interpreted in
	$\cat{GL}(\Z_2)$.
\end{theorem}

The same proof strategy as before works. Define functors between the circuit
prop obtained with mobit equations as well as control equations, and the prop
$\cat{GL}(\Z_2)_2$. Observe that $\cat{GL}(\Z_2)$ is a bipermutative
category, and thus $\cat{GL}(\Z_2)_2$ is a crop. Write $\cat
J$ for the crop generated by $x$ and $j$ such that $x^2 = \iid_1$, $j^3 =
\iid_1$, and $jxj = x$. Write $\ccat J\quo$ for the controlled prop out
of $\cat J$ further quotiented by \eqref{eq:mobit-missing}.

Similar to Gray code transpositions (see \secref{sec:gray}), we define Gray
code $J$'s to be:
\[
	\zeta_{n,i} =
	r_n^{-1} \circ (\iid_i \oplus J \oplus \iid_{2^n - 2 - i}) \circ r_n \text.
\]
for which we can also obtain an inductive formulation:
\begin{equation}\label{eq:ind-zeta}
	\zeta_{n+1,i} =
	\left\{
		\begin{array}{ll}
			\zeta_{n,i} \oplus \iid_{2^n} & \text{if } i < 2^n - 1\text, \\
			\symtimes_{2^n,2} \circ (\iid_{2^n} \oplus J \oplus \iid_{2^n-2}) \circ \symtimes_{2,2^n}
			& \text{if } i = 2^n - 1\text, \\
			\iid_{2^n} \oplus \left( \theta_{n,\chi(n)} \circ
			\zeta_{n,\chi(n)} \circ \theta_{n,\chi(n)} \right)
			& \text{otherwise.}
		\end{array}
	\right.
\end{equation}
with $\chi(n) = 2^{n+1} - 1 - i$.

This allows us the define the following prop morphisms, one direction obtained
through the prop and control structure, and the other because morphisms in
$\cat{GL}(\Z_2)$ are generated by $X$ and $J$, this time in Gray code:
$\theta_{n,i}$ and $\zeta_{n,i}$. In the definitions below, the order $<_\ell$
describes the lexicographic order.
\[
	\begin{array}{lrcllrcll}
		A_J \colon & \ccat J\quo & \to & \cat{GL}(\Z_2)_2& \\
		& n & \mapsto & 2^n & \\
		& f \in \ccat X & \mapsto & \alpha(f) & \\
		& j & \mapsto & J & \\
		& C^1(f) & \mapsto & \iid_{2^n} \oplus A_J(f) & \\
		& C^0(f) & \mapsto & A_J(f) \oplus \iid_{2^n} & \\
		& \iid_m + f & \mapsto & \iid_{2^m} \otimes A_J(f) &&&& \\[2ex]
		B_J \colon & \cat{GL}(\Z_2)_2 & \to & \ccat J\quo \\
		& 2^n & \mapsto & n \\
		& \theta_{n,i} & \mapsto & C^{w_n(i)}_{w'_n(i)}(x) \\[1.4ex]
		& \zeta_{n,i} & \mapsto & C^{w_n(i)}_{w'_n(i)}(j) & \text{if } h_n(i) <_\ell h_n(i+1) \text, \\[1.4ex]
		& \zeta_{n,i} & \mapsto & C^{w_n(i)}_{w'_n(i)}(j^2) & \text{otherwise.}
	\end{array}
\]
The functor $A_J$ is well-defined for the same reason as Lemma~\ref{lem:a-def}
and because \eqref{eq:mobit-missing} reflects \eqref{laf:ybqqq}. We then treat
the functor $B_J$ the same way as $\beta \colon \cat{Perm}_2 \to \ccat X$.

\begin{lemma}
	\label{lem:bj-wd}
	The functor $B_J \colon \cat{GL}(\Z_2)_2 \to \ccat J\quo$ is a well-defined
	crop morphism, and
	\[
		B_J(\iid_{2^n} \oplus f) = C^1(B_J(f)) \text,
		\qquad
		B_J(f \oplus \iid_{2^n}) = C^0(B_J(f)) \text,
	\]
	for all morphisms $f \colon 2^n \to 2^n$.
\end{lemma}
\begin{proof}
	In order to show that $B_J$ is well-defined, we need to prove that all
	equations that hold in $\cat{GL}(\Z_2)_2$ also hold after image by $B_J$.
	To do so, we work with the Lafont representation given earlier in the
	section. The equations \eqref{laf:xinv}, \eqref{laf:qinv}, and
	\eqref{laf:bit} are straightforward.
	Equation~\eqref{laf:yb} is the usual Yang-Baxter, therefore holds after
	applying $B_J$ for the same reason it does for $\beta$ (see
	Lemma~\ref{lem:beta-wd}). Similarly to Yang-Baxter for $\beta$,
	well-definedness at \eqref{laf:ybq} follows by induction, where the base
	case (on two mobits) is:
	\begin{align*}
		\tf{mobit-symm}
		=& \tf{mobit-symm-2} \\
		\stackrel{\ref{cprop:swap}+\ref{cprop:composition}}{=}& \tf{mobit-symm-3} \\
		\stackrel{\ref{cprop:composition}}{=}& \tf{mobit-symm-n}
	\end{align*}
	and the rest of the induction is direct because $\beta$ is well-defined. A
	similar reasoning applies to \eqref{laf:ybqqq}, however the base case of
	the induction cannot be proven simply by the rules of $\ccat J$, this is
	why \eqref{eq:mobit-missing} is added.
	Additionally, We use the same strategy than in Lemma~\ref{lem:beta-wd} to
	prove that $B_J$ is a prop morphism and maps the direct sums to control.
\end{proof}

Because $A_J$ and $B_J$ are well-defined crop morphisms and have inverse
images, the categories $\ccat J\quo$ and $\cat{GL}(\Z_2)_2$ are crop
isomorphic.

\subsection{Controlled-V}

An active area of research in quantum computing deals with determining
the precise resources necessary for universal quantum computing to
emerge~\cite{howardetal:contextuality, vedral:elusive,
 jozsalinden:entanglement}. As we have seen, in the classical case,
universality emerges from the simplest possible theory of control,
namely the one generated by a crop with a single involution
$NOT : 1 \to 1$. A natural question to ask is how this prop must be
extended to give a model of universal quantum computation.

Recent work~\cite{kaarsgaard:ctrlv} shows that a controlled-$V$ gate,
where $V$ is a \emph{square root} of the NOT gate (i.e., satisfies $V
\circ V = NOT$), suffices to perform universal quantum computing. This
has the surprising consequence that a model of universal quantum
computation can be generated from the theory of control for a single
operation of order $4$.

\begin{theorem}
 Let $(\cat{V},+,0,V^2)$ be the crop given by the prop presented by a
 single generator $V : 1 \to 1$ and relation $V^4 = \id$. Then there
 exists a prop morphism $\ccat{\cat{V}} \to \cat{Unitary}$ whose
 image is universal for quantum computing.
\end{theorem}
\begin{proof}
 Send $V$ to the usual $V$-gate $\tfrac{1}{2}\left(\begin{smallmatrix}
 1+i & 1-i \\ 1-i & 1+i
 \end{smallmatrix}\right)$, and its controlled variants to the usual
 (iterated) controlled $V$-gate.
\end{proof}

Here, universal is meant in the sense of computational universality~\cite{aharonov:toffolihadamard}.
While this gives a model of universal quantum computing, its
equational theory is not complete (i.e., the functor is not faithful);
this can be seen from the fact that the image of $V$ could have been
chosen to be the $S$-gate, which is not universal for quantum
computing even given arbitrary control (it generates a theory of
\emph{monomial matrices}). The precise relations needed to pin down
the theory of controlled-$V$ is an open question.

\subsection{Quantum circuits}

A complete equational theory for quantum circuits was an open question that was
solved only recently~\cite{clementetal:complete} with several later
improvements~\cite{clementetal:extensions, clementetal:minimal,
delorme2025control}. These papers already refer to some equations as
\emph{structural}. We show here that these equations are structural in a formal
and categorical sense: they are only about control, and follow directly from the
structure of a rig category.

Conjugation plays an important role in controlled quantum
circuits~\cite{delorme2025control}, because of the following equation.
\begin{equation}\label{eq:conj}
	\tf{conjugate} = \tf{conjugate-4}
\end{equation}
Note that this conjugation equation follows from our control theory as follows.
\begin{align*}
	\tf{conjugate}
	&\stackrel{\text{\ref{cprop:composition}}}{=}
	\tf{conjugate-2} \\
	&\stackrel{\text{\ref{cprop:commutativity}}}{=}
	\tf{conjugate-3} \\
	&\stackrel{\text{\ref{cprop:complementarity}}}{=}
	\tf{conjugate-4}
\end{align*}

Let $(\cat Z, +, 0)$ be the prop generated by $\alpha \colon 0 \to 0$ for $\alpha \in
\R$, and $Z(\alpha) \colon 1 \to 1$ for $\alpha \in \R$, and $H \colon 1 \to 1$
satisfying
\begin{align}
	\label{eq:2pi-phase}
	\tf{2pi-phase} &= \tf{2pi-phase-2}
	\\[1.5ex]
	\label{eq:phase-sum}
	\tf{phase-sum} &= \tf{phase-sum-2}
	\\[1.5ex]
	\label{eq:z-id}
	\tf{z-id} &= \tf{h-invo-2}
	\\[1.5ex]
	\label{eq:z-sum}
	\tf{z-sum-phase} &= \tf{z-sum-phase-2}
	\\[1.5ex]
	\label{eq:h-invo}
	\tf{h-invo} &= \tf{h-invo-2}
	\\[1.5ex]
	\label{eq:euler}
	\tf{euler} &= \tf{euler-2}
\end{align}
where \eqref{eq:euler} is the well-known Euler decomposition, in which
$\beta_0$, $\beta_1$, $\beta_2$ and $\beta_3$ can be computed from $\alpha_1$
and $\alpha_2$ deterministically. We choose the crop $(\cat Z, +, 0, HZ(\pi)H)$
and can now form its controlled prop $\ccat Z$. Interestingly, the prop $\ccat
Z$ is not immediately the prop of unitary operations: for a given $\alpha$, the
morphisms $C^1(\alpha)$ and $Z(\alpha)$ have the same semantics in unitaries
but are still different in $\ccat Z$. We need to quotient further with the
following equation:
\begin{equation}\label{eq:final-quantum}
	\tf{final-equation} = \tf{final-equation-2}
\end{equation}
and we write $\ccat Z\quo$ for this category. From this point, it is simple to
show that we have equivalent equations to the complete equational theory for
quantum circuits~\cite{delorme2025control}, and therefore $\ccat Z\quo$ is
isomorphic to the category of unitaries on qubits.

\begin{theorem}
	Equations \eqref{eq:2pi-phase}, \eqref{eq:phase-sum}, \eqref{eq:h-invo},
	\eqref{eq:euler} together with the control equations of
	Figure~\ref{fig:control-equations} and equation \eqref{eq:final-quantum},
	are sound and complete for quantum circuits.
\end{theorem}
\begin{proof}
	The set of generators is universal, and it suffices to show that all the required equations follow from ours. The circuit equations from the
	literature~\cite{delorme2025control} contain similar equations to
	\eqref{eq:2pi-phase}, \eqref{eq:phase-sum}, \eqref{eq:z-sum},
	\eqref{eq:h-invo}, and~\eqref{eq:euler}, given that equation
	\eqref{eq:final-quantum} holds. Moreover, we have proven that controlled
	conjugation \eqref{eq:conj} holds in our control theory. This is enough to
	cover the set of equations for quantum circuits~\cite{delorme2025control}.
\end{proof}

In fact, the same general strategy works to obtain a complete
equational theory for quantum circuits formed from certain discrete
gate sets by exploiting the rig structure shown in
\cite{carette2024easy}. Let $(\cat{\Pi},+,0)$ be the prop generated by
$\omega : 0 \to 0$, $V : 1 \to 1$, and $S : 1 \to 1$ satisfying
\begin{align}
 \tf{omega} &= \tf{2pi-phase-2} \label{eq:omega} \\[1.5ex]
 \tf{v-order-4} &= \tf{h-invo-2} \label{eq:v-order-4} \\[1.5ex]
 \tf{SVS} &= \tf{VSV} \label{eq:svs}
\end{align}
Choose now the crop $(\cat{\Pi},+,0,V^2)$, form its controlled
prop $\ccat{\Pi}$, and quotient it further by the equation
\begin{equation}
	\tf{s} = \tf{ctrl-i} \label{eq:ctrl-i}
\end{equation}
to obtain the category $\ccat \Pi\quo$. Recall that the \emph{Gaussian
Clifford+$T$} gate set~\cite{amy2020number} consists of the gates $S$,
$K$, $X$, $CX$, and $CCX$, where $K = \tfrac{1-i}{\sqrt{2}}H$, and $H$
is the usual Hadamard gate. These gates are expressible using the
generators $S$,$V$, $\omega$ as well as the control functors: $K =
\omega^6 \bullet SVS$ (where $\bullet$ denotes the usual scalar
multiplication in symmetric monoidal categories), $X = V^2$ is the
distinguished involution, $CX = C^1(X)$, and $CCX = C^1(CX)$. To
express Clifford+$T$, we additionally use $H = \omega^7 \bullet SVS$
and $T = C^1(\omega)$.

This equational theory can express a number of more restricted
theories, including ones that have known presentations (notably
Clifford~\cite{selinger:clifford}, unitaries generated by Gaussian
Clifford+$T$~\cite{bianselinger:unitaries}, and $2$-qubit
Clifford+$T$~\cite{bianselinger:cliffordt}). It was shown
previously~\cite{carette2024easy,caretteetal:free} that the equations
of these theories are implied by the ones above, making the larger
theory (equationally) sound and complete for these.
\begin{theorem}
 Equations \eqref{eq:omega}, \eqref{eq:v-order-4}, \eqref{eq:svs},
 together with the control equations of
 Figure~\ref{fig:control-equations} and equation
 \eqref{eq:ctrl-i}, are sound and complete for Clifford circuits,
 Clifford+$T$ circuits with $\le 2$ qubits, and Gaussian Clifford+$T$ circuits.
\end{theorem}
\begin{proof}
 Establishing the required completeness results in
 \cite{carette2024easy,caretteetal:free} only required $\omega^8 =
 \id$ (which follows by \eqref{eq:omega}), $V^2 = \symplus_{1,1}$
 (which follows by \eqref{eq:v-order-4} and that fact that $V^2$ is
 the chosen involution), and $SVS = VSV$ where $S = \id \oplus
 \omega^2$ (which follow from \eqref{eq:svs} and \eqref{eq:ctrl-i}),
 as well as the axioms of rig categories, which are implied by the
 control equations of Figure~\ref{fig:control-equations} by
 Theorem~\ref{th:cp-iso}.
\end{proof}

\section{Conclusion}
\label{sec:conclusion}

Our results raise several interesting questions, that we leave to future work.
\begin{itemize}
	\item The control equations of \Cref{fig:control-equations} have natural interpretations that seem to indicate they are independent. But is this set of equations indeed minimal? Can we find models that satisfy all but one of the control equations?

	\item Several circuit optimisations in the literature rely on auxiliary
		wires in a circuit to enact controlled
		operations~\cite{clementetal:minimal,clementetal:extensions}, and other
		works also make use of ancillas to obtain complete equational theories
		for CNOTs~\cite{cockettetal:cnot} and
		Toffolis~\cite{cockettcomfort:tof}. How do these relate to rigged
		crops?

	\item There is no quantum circuit implementation of a controlled unitary where the unitary is a black box input~\cite{araujoetal:nocontrol}. Yet there are many physical implementations bypassing the strict confines of this no-go theorem~\cite{friisetal:unknowncontrol,bisioetal:conditional}. The latter rely on the ability to identify subspaces of qubits by adding auxiliary dimensions, harmonising with our results. Is there a version of the construction of \Cref{def:cprop} that allows for auxiliary wires?

	\item We have isolated the control aspects from the rest of the computation. Can this be used to analyse causal dependencies within the data flow of the computation?

	\item We have worked out examples adding computational control to several props such as a single $X$ gate, a single $\sqrt{X}$ gate, and an $X$ gate together with a Hadamard gate. What does the controlled prop look like for other concrete base circuit theories?

	\item The control equations of \Cref{fig:control-equations} implicitly assume only two possibilities on each wire: positive and negative control. Can we extend the theory from two-valued data like bits or qubits to data with multiple values like qutrits?

	\item We have worked with bipermutative categories, where the fact that $\oplus$ is symmetric corresponds to the fact $x^2=1$ that the NOT gate is an involution. Are there interesting theories when we generalise to $\oplus$ being merely braided?

	\item Boolean functions are fundamental to bounded decision diagrams and SAT-solving. Do our results imply anything practically useful in those applications?
\end{itemize}

\subsection*{Acknowledgements}

The authors would like to thank the members of the quantum programming group,
specifically Wang Fang, and of the category theory group at Edinburgh for their
support and feedback. We extend our thanks to Robert Booth, Kostia Chardonnet,
Noé Delorme, Nicolas Heurtel, Simon Perdrix, and Peter Selinger for comments
and fruitful discussions that helped us improve the paper. We also thank the
anonymous reviewers for their comments and suggestions. This research was
funded by the Engineering and Physical Sciences Research Council (EPSRC) under
project EP/X025551/1 “Rubber DUQ: Flexible Dynamic Universal Quantum
programming”. Robin Kaarsgaard was supported by Sapere Aude: DFF-Research
Leader grant 5251-00024B.

\bibliography{ref}

\appendix
\onecolumn

\section{Definitions}
\label{app:defs}
\begin{definition}[Bipermutative category]
	A \emph{bipermutative category} $(\cat C, \otimes, I, \oplus, O)$ is a
	category such that $(C, \otimes, I)$ and $(C, \oplus, O)$ are strict
	symmetric monoidal categories, and the following conditions are satisfied:
	\begin{itemize}
		\item for all objects $A$, we have $O \otimes A = O = A \otimes O$; for
			all morphisms $f \colon A \to B$, we have $f \otimes \iid_O =
			\iid_O = \iid_O \otimes f$;
		\item for all objects $A, B, C$ and morphisms $f, g, h$, we have that
			\begin{equation}\label{eq:biper-eq-types}
				(A \oplus B) \otimes C = (A \otimes C) \oplus (B \otimes C)
			\end{equation}
			and
			\begin{equation}\label{eq:biper-eq}
				(f \oplus g) \otimes h = (f \otimes h) \oplus (g \otimes h)
			\end{equation}
			and the following diagram commutes:
			\begin{equation}\label{eq:biper-sym}
				\begin{tikzcd}
					(A \oplus B) \otimes C & (A \otimes C) \oplus (B \otimes C) \\
					(B \oplus A) \otimes C & (B \otimes C) \oplus (A \otimes C)
					\arrow["=", from=1-1, to=1-2]
					\arrow["=", from=2-1, to=2-2]
					\arrow["\symplus \otimes \iid", from=1-1, to=2-1]
					\arrow["\symplus", from=1-2, to=2-2]
				\end{tikzcd}
			\end{equation}
		\item if we define a natural left distributor $\delta$ as:
			\begin{equation}\label{eq:biper-dist}
				\begin{tikzcd}[column sep=small]
					A \otimes (B \oplus C) & (B \oplus C) \otimes A & (B
					\otimes A) \oplus (C \otimes A) & (A \otimes B) \oplus (A
					\otimes C)
					\arrow["s", from=1-1, to=1-2]
					\arrow["=", from=1-2, to=1-3]
					\arrow["s \oplus s", from=1-3, to=1-4]
				\end{tikzcd}
			\end{equation}
			then the following diagram commutes for all objects $A, B, C, D$:
			\begin{equation}\label{eq:biper-coh}
				\begin{tikzcd}
					(A \oplus B) \otimes (C \oplus D) & (A \otimes (C \oplus D)) \oplus (B \otimes (C \oplus D)) \\
					((A \oplus B) \otimes C) \oplus ((A \oplus B) \otimes D) & (A
					\otimes C) \oplus (A \otimes D) \oplus (B \otimes C) \oplus (B
					\otimes D) \\
					& (A \otimes C) \oplus (B \otimes C) \oplus (A \otimes D) \oplus (B \otimes D)
					\arrow["=", from=1-1, to=1-2]
					\arrow["\delta", from=1-1, to=2-1]
					\arrow["\delta \oplus \delta", from=1-2, to=2-2]
					\arrow["=", from=2-1, to=3-2]
					\arrow["\iid \oplus \symplus \oplus \iid", from=2-2, to=3-2]
				\end{tikzcd}
			\end{equation}
	\end{itemize}
	A \emph{strict bipermutative functor} is a functor which is a strict symmetric
	monoidal functor with respect to both symmetric monoidal structures $(C,
	\otimes, I)$ and $(C, \oplus, O)$.
\end{definition}

\end{document}

%% file: figures/quantum-circuits.tikzstyles
\tikzstyle{gate}=[fill=white, draw=black, shape=rectangle, minimum height=0.43cm, minimum width=0.43cm, inner sep=0.1em]
\tikzstyle{control}=[fill=black, draw=black, shape=circle, scale=0.38]
\tikzstyle{not}=[shape=circle, path picture={ 
\draw[black](path picture bounding box.north) -- (path picture bounding box.south) (path picture bounding box.west) -- (path picture bounding box.east);
}, draw=black, scale=.8]
\tikzstyle{wcontrol}=[fill=white, draw=black, shape=circle, scale=0.35]
\tikzstyle{empty}=[fill=white, draw=black, shape=rectangle, inner sep=0.4em, emptyborder]
\tikzstyle{globalphase}=[fill=white, draw=black, inner sep=0.15em, shape=rounded rectangle, minimum height=0.4cm]
\tikzstyle{ancilla}=[fill=black, draw=black, shape=rectangle, minimum width=0.01cm, minimum height=0.25cm, inner sep=0.01em]
\tikzstyle{ground}=[fill=white, path picture={\draw[black](-1.5mm,0)--(-0.6mm,0);\draw[black,thick](-0.6mm,-1.75mm)--(-0.6mm,1.75mm) (0mm,-0.9mm)--(0mm,0.9mm) (0.6mm,-0.5mm)--(0.6mm,0.5mm);}, minimum width=0.1mm, draw=none, outer sep=0pt]
\tikzstyle{gate22}=[fill=white, draw=black, shape=rectangle, minimum height=.68cm, minimum width=0.6cm]
\tikzstyle{void}=[shape=rectangle, minimum height=0.5cm]
\tikzstyle{gate44}=[fill=white, draw=black, shape=rectangle, minimum height=1.43cm, minimum width=0.5cm]
\tikzstyle{divider}=[fill={rgb,255: red,220; green,220; blue,220}, draw=black, shape border rotate=90, regular polygon, regular polygon sides=3, inner sep=1.5pt, rounded corners=0.5mm]
\tikzstyle{gatherer}=[fill={rgb,255: red,220; green,220; blue,220}, draw=black, shape border rotate=-90, regular polygon, regular polygon sides=3, inner sep=1.5pt, rounded corners=0.5mm]
\tikzstyle{hyperedge}=[fill=white, draw=black, shape=rectangle, rounded corners=0.1cm, minimum height=.6cm, minimum width=.6cm]
\tikzstyle{square}=[fill=white, draw=black, shape=rectangle, minimum height=0.20cm, minimum width=0.20cm, inner sep=0.1em, thick]
\tikzstyle{gphase}=[rounded rectangle, rounded rectangle arc length=120, fill={zx_grey}, inner sep=2pt, font={\tiny\boldmath}, label distance=1mm, fill opacity=.8, text opacity=1, tikzit category=ZX]
\tikzstyle{customcontrol}=[fill=white, draw=black, inner sep=0.1em, shape=rounded rectangle, minimum height=0.2cm]

\tikzstyle{emptyborder}=[-, dash pattern=on 0.16em off 0.16em on 0.16em off 0.16em on 0.16em off 0em]
\tikzstyle{etc}=[-, draw=black, densely dashed, thick]
\tikzstyle{greyetc}=[-, draw={rgb,255: red,161; green,161; blue,161}, densely dashed, thick]
\tikzstyle{dots}=[-, dotted, draw=black, thick]
\tikzstyle{big}=[-, thick]
\tikzstyle{register}=[-, double]
\tikzstyle{grey}=[-, draw={rgb,255: red,161; green,161; blue,161}]
\tikzstyle{border}=[-, fill=white]

%% file: commands.tex
\makeatletter
\newcommand\Label[1]{\refstepcounter{equation}(\theequation)\ltx@label{#1}}
\makeatother

\usepackage{color}
\def\bR{\begin{color}{red}}
\def\bB{\begin{color}{blue}}
\def\bM{\begin{color}{magenta}}
\def\bC{\begin{color}{cyan}}
\def\bW{\begin{color}{white}}
\def\bBl{\begin{color}{black}}
\def\bG{\begin{color}{green}}
\def\bY{\begin{color}{yellow}}
\def\e{\end{color}\xspace}
\newcommand{\bit}{\begin{itemize}}
\newcommand{\eit}{\end{itemize}\par\noindent}
\newcommand{\ben}{\begin{enumerate}}
\newcommand{\een}{\end{enumerate}\par\noindent}
\newcommand{\beq}{\begin{equation}}
\newcommand{\eeq}{\end{equation}\par\noindent}
\newcommand{\beqa}{\begin{eqnarray*}}
\newcommand{\eeqa}{\end{eqnarray*}\par\noindent}
\newcommand{\beqn}{\begin{eqnarray}}
\newcommand{\eeqn}{\end{eqnarray}\par\noindent}

\newcommand{\alt}{~\mid~}

\newcommand{\cat}[1]{\mathbf{#1}}
\newcommand{\ccat}[1]{\mathbf{c #1}}
\newcommand{\kcat}[1]{\mathbf{\widetilde{#1}}}
\newcommand{\scat}[1]{\mathbf{\ovl{#1}}}
\newcommand{\tcat}[1]{\mathbf{\ovl{#1}}_2}

\newcommand{\symtimes}{\gamma^\otimes}
\newcommand{\symplus}{\gamma^\oplus}

\newcommand{\id}{\mathrm{id}}

\newcommand{\Z}{{\mathbb Z}}

\newcommand{\R}{{\mathbb R}}

\newcommand{\defeq}{\stackrel{\textrm{{\scriptsize def}}}{=}}

\newcommand{\quo}{_{/\simeq}}

\newcommand\embed\hookrightarrow

\newcommand\natto\Rightarrow

\newcommand{\abs}[1]{\left\vert #1 \right\vert}

\newcommand{\iid}{\mathrm{id}}

\newcommand{\den}[1]{\left\llbracket #1 \right\rrbracket}
\newcommand\sem\den

\newcommand{\ovl}[1]{\overline{#1}}

\DeclareFontFamily{U}{min}{}
\DeclareFontShape{U}{min}{m}{n}{<-> udmj30}{}

\newcommand{\secref}[1]{\S \ref{#1}}

%% file: control.bbl
\begin{thebibliography}{10}

\bibitem{abdessaied2016complexity}
N.~Abdessaied, M.~Amy, R.~Drechsler, and M.~Soeken.
\newblock Complexity of reversible circuits and their quantum implementations.
\newblock {\em Theoretical Computer Science}, 618:85--106, 2016.

\bibitem{aharonov:toffolihadamard}
D.~Aharonov.
\newblock A simple proof that {T}offoli and {H}adamard are quantum universal.
\newblock arXiv preprint, 2003.
\newblock \href {https://arxiv.org/abs/arXiv:quant-ph/0301040}
  {\path{arXiv:arXiv:quant-ph/0301040}}.

\bibitem{amy2020number}
M.~Amy, A.~N. Glaudell, and N.~J. Ross.
\newblock Number-theoretic characterizations of some restricted clifford+{T}
  circuits.
\newblock {\em Quantum}, 4:252, 2020.

\bibitem{araujoetal:nocontrol}
M.~Ara{\'u}jo, F.~Adrien, F.~Costa, and {\v{C}}.~Brukner.
\newblock Quantum circuits cannot control unknown operations.
\newblock {\em New Journal of Physics}, 16(9):093026, 2014.
\newblock \href {https://doi.org/10.1088/1367-2630/16/9/093026}
  {\path{doi:10.1088/1367-2630/16/9/093026}}.

\bibitem{baezetal:props}
J.~C. Baez, B.~Coya, and F.~Rebro.
\newblock Props in network theory.
\newblock {\em Theory and Applications of Categories}, 33(25):727--783, 2018.
\newblock \href {https://arxiv.org/abs/arXiv:1707.08321}
  {\path{arXiv:arXiv:1707.08321}}, \href
  {https://doi.org/10.70930/tac/mgee83ch} {\path{doi:10.70930/tac/mgee83ch}}.

\bibitem{balaucaarusaoie:simulating}
S.~Balauca and A.~Arusoaie.
\newblock Efficient constructions for simulating multi controlled quantum
  gates.
\newblock In {\em International Conference on Computer Science}, volume 13353
  of {\em Lecture Notes in Computer Science}, pages 179--194. Springer, 2022.
\newblock \href {https://doi.org/10.1007/978-3-031-08760-8\_16}
  {\path{doi:10.1007/978-3-031-08760-8\_16}}.

\bibitem{barencoetal:gates}
A.~Barenco, C.~H. Bennett, R.~Cleve, D.~DiVincenzo, N.~Margolus, P.~Shor,
  T.~Sleator, J.~A. Smolin, and H.~Weinfurter.
\newblock Elementary gates for quantum computation.
\newblock {\em Physical Review A}, 5:3457--3467, 1995.
\newblock \href {https://doi.org/10.1103/PhysRevA.52.3457}
  {\path{doi:10.1103/PhysRevA.52.3457}}.

\bibitem{bianselinger:unitaries}
X.~Bian and P.~Selinger.
\newblock {Generators and relations for $U_n(\mathbb{Z}[\frac{1}{2},i]$)}.
\newblock In {\em Quantum Physics and Logic}, volume 343 of {\em Electronic
  Proceedings in Theoretical Computer Science}, pages 145--164, 2021.
\newblock \href {https://doi.org/10.4204/EPTCS.343.8}
  {\path{doi:10.4204/EPTCS.343.8}}.

\bibitem{bianselinger:cliffordt}
X.~Bian and P.~Selinger.
\newblock Generators and relations for 2-qubit {Clifford+T} operators.
\newblock {\em Electronic Proceedings in Theoretical Computer Science},
  394:13--28, 2023.
\newblock \href {https://doi.org/10.4204/eptcs.394.2}
  {\path{doi:10.4204/eptcs.394.2}}.

\bibitem{bisioetal:conditional}
A.~Bisio, M.~{Dall'Arno}, and P.~Perinotti.
\newblock Quantum conditional operations.
\newblock {\em Physical Review A}, 94:022340, 2016.
\newblock \href {https://doi.org/10.1103/PhysRevA.94.022340}
  {\path{doi:10.1103/PhysRevA.94.022340}}.

\bibitem{bonchisobocinskizanasi:hopf}
F.~Bonchi, B.~Soboci{\'n}ski, and F.~Zanasi.
\newblock Interacting {H}opf algebras.
\newblock {\em Journal of Pure and Applied Algebra}, 221(1):144--184, 2017.
\newblock \href {https://doi.org/10.1016/j.jpaa.2016.06.002}
  {\path{doi:10.1016/j.jpaa.2016.06.002}}.

\bibitem{caretteetal:free}
J.~Carette, C.~Heunen, R.~Kaarsgaard, N.~J. Ross, and A.~Sabry.
\newblock Free quantum computing.
\newblock {\em Proceedings of the National Academy of Sciences},
  123(8):e2510881123, 2026.
\newblock \href {https://doi.org/10.1073/pnas.2510881123}
  {\path{doi:10.1073/pnas.2510881123}}.

\bibitem{carette2023quantum}
J.~Carette, C.~Heunen, R.~Kaarsgaard, and A.~Sabry.
\newblock The quantum effect: A recipe for quantum-{$\Pi$}.
\newblock In {\em Proceedings of the {ACM} on Programming Languages}, volume~8,
  pages 1--29, 2024.
\newblock \href {https://doi.org/10.1145/3674625} {\path{doi:10.1145/3674625}}.

\bibitem{carette2024easy}
J.~Carette, C.~Heunen, R.~Kaarsgaard, and A.~Sabry.
\newblock With a few square roots, quantum computing is as easy as {$\Pi$}.
\newblock {\em Proceedings of the {ACM} on Programming Languages}, 8, 2024.
\newblock \href {https://doi.org/10.1145/3632861} {\path{doi:10.1145/3632861}}.

\bibitem{carette2016computing}
J.~Carette and A.~Sabry.
\newblock Computing with semirings and weak rig groupoids.
\newblock In {\em European Symposium on Programming}, volume 9632 of {\em
  Lecture Notes in Theoretical Computer Science}, pages 123--148. Springer,
  2016.
\newblock \href {https://doi.org/10.1007/978-3-662-49498-1\_6}
  {\path{doi:10.1007/978-3-662-49498-1\_6}}.

\bibitem{chenetal:control}
Z.~Chen, W.~Liu, Y.~Ma, W.~Sun, R.~Wang, H.~Wang, H.~Xu, G.~Xue, H.~Yan,
  Z.~Yang, J.~Ding, Y.~Gao, F.~Li, Y.~Zhang, Z.~Zhang, Y.~Jin, H.~Yu, J.~Chen,
  and F.~Yan.
\newblock Efficient implementation of arbitrary two-qubit gates using unified
  control.
\newblock {\em Nature Physics}, 21:1489--1496, 2025.
\newblock \href {https://doi.org/10.1038/s41567-025-02990-x}
  {\path{doi:10.1038/s41567-025-02990-x}}.

\bibitem{clementetal:minimal}
A.~Cl\'{e}ment, N.~Delorme, and S.~Perdrix.
\newblock Minimal equational theories for quantum circuits.
\newblock In {\em Logic in Computer Science}, pages 27:1--27:14. ACM/IEEE,
  2024.
\newblock \href {https://doi.org/10.1145/3661814.3662088}
  {\path{doi:10.1145/3661814.3662088}}.

\bibitem{clementetal:extensions}
A.~Cl\'{e}ment, N.~Delorme, S.~Perdrix, and R.~Vilmart.
\newblock Quantum circuit completeness: Extensions and simplifications.
\newblock In {\em Computer Science Logic}, volume 288 of {\em Leibniz
  International Proceedings in Informatics (LIPIcs)}, pages 20:1--20:23, 2024.
\newblock \href {https://doi.org/10.4230/LIPIcs.CSL.2024.20}
  {\path{doi:10.4230/LIPIcs.CSL.2024.20}}.

\bibitem{clementetal:complete}
A.~Cl\'{e}ment, N.~Heurtel, S.~Mansfield, S.~Perdrix, and B.~Valiron.
\newblock A complete equational theory for quantum circuits.
\newblock In {\em Logic in Computer Science}, pages 1--13. ACM/IEEE, 2023.
\newblock \href {https://doi.org/10.1109/LICS56636.2023.10175801}
  {\path{doi:10.1109/LICS56636.2023.10175801}}.

\bibitem{clementperdrix:pbs}
A.~Cl\'{e}ment and S.~Perdrix.
\newblock {PBS}-calculus: A graphical language for coherent control of quantum
  computations.
\newblock In {\em Mathematical Foundations of Computer Science}, volume 170 of
  {\em Leibniz International Proceedings in Informatics (LIPIcs)}, pages
  24:1--24:14, 2020.
\newblock \href {https://doi.org/10.4230/LIPIcs.MFCS.2020.24}
  {\path{doi:10.4230/LIPIcs.MFCS.2020.24}}.

\bibitem{clementetal:coherentcontrol}
A.~Cl\'{e}ment and S.~Perdrix.
\newblock Resource optimisation of coherently controlled quantum computations
  with the pbs-calculus.
\newblock In {\em Mathematical Foundations of Computer Science}, volume 241 of
  {\em Leibniz International Proceedings in Informatics (LIPIcs)}, pages
  36:1--36:15, 2022.
\newblock \href {https://doi.org/10.4230/LIPIcs.MFCS.2022.36}
  {\path{doi:10.4230/LIPIcs.MFCS.2022.36}}.

\bibitem{cockettcomfort:tof}
J.R.B. Cockett and C.~Comfort.
\newblock The category {TOF}.
\newblock {\em Electronic Proceedings in Theoretical Computer Science},
  287:67–84, 2019.
\newblock URL: \url{http://dx.doi.org/10.4204/EPTCS.287.4}, \href
  {https://doi.org/10.4204/eptcs.287.4} {\path{doi:10.4204/eptcs.287.4}}.

\bibitem{cockettetal:cnot}
R.~Cockett, C.~Comfort, and P.~Srinivasan.
\newblock The category cnot.
\newblock {\em Electronic Proceedings in Theoretical Computer Science},
  266:258–293, 2018.
\newblock URL: \url{http://dx.doi.org/10.4204/EPTCS.266.18}, \href
  {https://doi.org/10.4204/eptcs.266.18} {\path{doi:10.4204/eptcs.266.18}}.

\bibitem{delorme2025control}
N.~Delorme and S.~Perdrix.
\newblock Diagrammatic reasoning with control as a constructor, applications to
  quantum circuits, 2025.
\newblock arXiv:2508.21756. To appear in: Proceedings of FoSSaCS'26.
\newblock URL: \url{https://arxiv.org/abs/2508.21756}, \href
  {https://arxiv.org/abs/2508.21756} {\path{arXiv:2508.21756}}.

\bibitem{erbele:control}
J.~M. Erbele.
\newblock {\em Categories in control: applied {PROP}s}.
\newblock PhD thesis, University of California, Riverside, 2016.

\bibitem{fangetal:hadamardpi}
W.~Fang, C.~Heunen, and R.~Kaarsgaard.
\newblock Hadamard-{$\Pi$}: Equational quantum programming.
\newblock arXiv:2506.06835, 2025.
\newblock URL: \url{https://arxiv.org/abs/2506.06835}.

\bibitem{friisetal:unknowncontrol}
N.~Friis, V.~Dunjko, W.~D\"ur, and H.~J. Briegel.
\newblock Implementing quantum control for unknown subroutines.
\newblock {\em Physical Review A}, 89:030303, 2014.
\newblock \href {https://doi.org/10.1103/PhysRevA.89.030303}
  {\path{doi:10.1103/PhysRevA.89.030303}}.

\bibitem{ghicajung:digitalcircuits}
D.~R. Ghica and A.~Jung.
\newblock Categorical semantics of digital circuits.
\newblock In {\em Formal Methods in Computer-Aided Design}, pages 41--48, 2016.
\newblock \href {https://doi.org/10.1109/FMCAD.2016.7886659}
  {\path{doi:10.1109/FMCAD.2016.7886659}}.

\bibitem{heunen2022information}
C.~Heunen and R.~Kaarsgaard.
\newblock Quantum information effects.
\newblock {\em Proceedings of the {ACM} on Programming Languages}, 6, 2022.
\newblock \href {https://doi.org/10.1145/3498663} {\path{doi:10.1145/3498663}}.

\bibitem{howardetal:contextuality}
M.~Howard, J.~Wallman, V.~Veitch, and J.~Emerson.
\newblock Contextuality supplies the `magic' for quantum computation.
\newblock {\em Nature}, 510:351--355, 2014.

\bibitem{jozsalinden:entanglement}
R.~Jozsa and N.~Linden.
\newblock On the role of entanglement in quantum-computational speed-up.
\newblock {\em Proceedings of the Royal Society of London A}, 459:2011--2032,
  2003.

\bibitem{kaarsgaard:ctrlv}
R.~Kaarsgaard.
\newblock All you need is controlled-{V}: universality of a standard two-qubit
  gate by catalytic embedding.
\newblock arXiv:2509.07578, 2025.

\bibitem{lafont2003}
Y.~Lafont.
\newblock Towards an algebraic theory of boolean circuits.
\newblock {\em Journal of Pure and Applied Algebra}, 184(2):257--310, 2003.
\newblock \href {https://doi.org/10.1016/S0022-4049(03)00069-0}
  {\path{doi:10.1016/S0022-4049(03)00069-0}}.

\bibitem{laplaza1972coherence}
M.~L. Laplaza.
\newblock Coherence for distributivity.
\newblock In G.~M. Kelly, M.~Laplaza, G.~Lewis, and Saunders Mac~Lane, editors,
  {\em Coherence in Categories}, pages 29--65. Springer, 1972.

\bibitem{lirossselinger:orthogonal}
S.~M. Li, N.~J. Ross, and P.~Selinger.
\newblock Generators and relations for the group
  $\mathrm{O}_n\left(\mathbb{Z}\left[{1}/{2}\right]\right)$.
\newblock In {\em Quantum Physics and Logic}, volume 343 of {\em Electronic
  Proceedings in Theoretical Computer Science}, pages 210--264, 2021.
\newblock \href {https://doi.org/10.4204/eptcs.343.11}
  {\path{doi:10.4204/eptcs.343.11}}.

\bibitem{maclane:props}
S.~{Mac Lane}.
\newblock Categorical algebra.
\newblock {\em Bulletin of the American Mathematical Society}, 71:40--106,
  1965.

\bibitem{may}
J.~P. May.
\newblock {\em $E_\infty$ Ring Spaces and $E_\infty$ Ring Spectra}, volume 577
  of {\em Lecture Notes in Mathematics}.
\newblock Springer-Verlag, 1977.

\bibitem{nielsenchuang}
M.~A. Nielsen and I.~L. Chuang.
\newblock {\em Quantum Computation and Quantum Information: 10th Anniversary
  Edition}.
\newblock Cambridge University Press, 2010.
\newblock \href {https://doi.org/10.1017/CBO9780511976667}
  {\path{doi:10.1017/CBO9780511976667}}.

\bibitem{perketauyang:yangbaxter}
J.~H.~H. Perk and H.~Au-Yang.
\newblock Yang--{B}axter equations.
\newblock In {\em Encyclopedia of Mathematical Physics}, pages 465--473.
  Academic Press, 2006.
\newblock \href {https://doi.org/10.1016/B0-12-512666-2/00191-7}
  {\path{doi:10.1016/B0-12-512666-2/00191-7}}.

\bibitem{schumacherwestmoreland:modal}
B.~Schumacher and M.~D. Westmoreland.
\newblock Modal quantum theory.
\newblock In {\em Quantum Physics and Logic}, pages 145--151, 2010.

\bibitem{selinger:clifford}
P.~Selinger.
\newblock Generators and relations for n-qubit clifford operators.
\newblock {\em Logical Methods in Computer Science}, 11, 2015.

\bibitem{sharmaachour:optimizing}
R.~Sharma and S.~Archour.
\newblock Optimizing ancilla-based quantum circuits with {SPARE}.
\newblock {\em Proceedings of the ACM on Programming Languages}, 9, 2025.
\newblock \href {https://doi.org/10.1145/3729253} {\path{doi:10.1145/3729253}}.

\bibitem{shendeetal:synthesis}
V.~V. Shende, S.~S. Bullock, and I.~L. Markov.
\newblock Synthesis of quantum logic circuits.
\newblock In {\em Asia and South Pacific Design Automation Conference}, pages
  272--275. ACM, 2005.
\newblock \href {https://doi.org/10.1145/1120725.1120847}
  {\path{doi:10.1145/1120725.1120847}}.

\bibitem{shi:aharonov}
Y.~Shi.
\newblock Both {T}offoli and controlled-{NOT} need little help to do universal
  quantum computing.
\newblock {\em Quantum Information and Computation}, 3(1):84--92, 2003.

\bibitem{sleatorweinfurter}
T.~Sleator and H.~Weinfurter.
\newblock Realizable universal quantum logic gates.
\newblock {\em Physical Review Letters}, 74:4087--4090, 1995.
\newblock \href {https://doi.org/10.1103/PhysRevLett.74.4087}
  {\path{doi:10.1103/PhysRevLett.74.4087}}.

\bibitem{thomsenkaarsgaardsoeken:ricercar}
M.~K. Thomsen, R.~Kaarsgaard, and M.~Soeken.
\newblock Ricercar: a language for describing and rewriting reversible circuits
  with ancillae and its permutation semantics.
\newblock In {\em Reversible Computing}, pages 200--215, 2015.
\newblock \href {https://doi.org/10.1007/978-3-319-20860-2\_13}
  {\path{doi:10.1007/978-3-319-20860-2\_13}}.

\bibitem{toffoli}
T.~Toffoli.
\newblock Reversible computing.
\newblock In {\em International Colloquium on Automata, Languages, and
  Programming}, Lecture Notes in Computer Science, pages 632--644. Springer,
  1980.
\newblock \href {https://doi.org/10.1007/3-540-10003-2\_104}
  {\path{doi:10.1007/3-540-10003-2\_104}}.

\bibitem{vedral:elusive}
V.~Vedral.
\newblock The elusive source of quantum speedup.
\newblock {\em Foundations of Physics}, 40:1141--1154, 2010.

\bibitem{vollmer:circuits}
H.~Vollmer.
\newblock {\em Introduction to circuit complexity}.
\newblock Springer, 1999.
\newblock \href {https://doi.org/10.1007/978-3-662-03927-4}
  {\path{doi:10.1007/978-3-662-03927-4}}.

\bibitem{bimonoidal-book}
D.~Yau.
\newblock {\em Bimonoidal Categories, $E_n$-Monoidal Categories, and Algebraic
  $K$-Theory: Volume I: Symmetric Bimonoidal Categories and Monoidal
  Bicategories}, volume 283 of {\em Mathematical Surveys and Monographs}.
\newblock American Mathematical Society, 2024.

\bibitem{yuduanying:toffoli}
N.~Yu, R.~Duan, and M.~Ying.
\newblock Five two-qubit gates are necessary for implementing the {T}offoli
  gate.
\newblock {\em Physical Review A}, 88:010304, 2013.
\newblock \href {https://doi.org/10.1103/physreva.88.010304}
  {\path{doi:10.1103/physreva.88.010304}}.

\end{thebibliography}
